\documentclass[12pt]{article}

\usepackage{hyperref}
\usepackage{latexsym,amsmath}
\usepackage{enumerate}
\usepackage{amsfonts}
\usepackage{amssymb}
\usepackage{amsthm}
\usepackage{latexsym}
\usepackage{fixmath}
\usepackage{mathtools}
\usepackage[usenames,dvipsnames]{color}
\usepackage{mathdots}
\usepackage{nicematrix}
\usepackage[capitalize]{cleveref}
\usepackage{listings}
\usepackage{aligned-overset} 
\usepackage{tikz}
\usetikzlibrary{angles,quotes}
\def\myrad{1.8}
 
\usepackage{graphicx}
\usepackage{comment}

\usepackage[T1]{fontenc}
\usepackage{fourier}
\usepackage{bbm}
\usepackage{color}
\usepackage{fullpage}

\numberwithin{equation}{section}

\renewcommand{\vec}[1]{\boldsymbol{#1}}
\newcommand\cD{\mathcal D}
\newcommand\cU{\mathcal U}
\newcommand\cV{\mathcal V}
\newcommand\vu{\vec u}
\newcommand\vw{\vec w}

\newcommand\eps{\epsilon}
\newcommand{\bck}[1]{\left\langle{#1}\right\rangle}
\newcommand\norm[1]{\left\|{#1}\right\|}
\newcommand\abs[1]{\left|{#1}\right|}

\newtheorem{definition}{Definition}[section]
\newtheorem{claim}[definition]{Claim}
\newtheorem{remark}[definition]{Remark}
\newtheorem{theorem}[definition]{Theorem}
\newtheorem{lemma}[definition]{Lemma}

\DeclareMathOperator{\sign}{sign}
\DeclareMathOperator{\slerp}{slerp}
\DeclareMathOperator{\NC}{NC}

\usepackage{graphicx} % Required for inserting images
\usepackage{caption}
\usepackage{subcaption}

\usepackage[style=alphabetic, maxbibnames=99, maxalphanames=4,
minalphanames=4,url=false,isbn=false,doi=false,eprint=false]{biblatex}
\title{Inevitability of Polarization in Geometric Opinion Exchange}

\author{
Abdou Majeed Alidou\footnote{AIMS Rwanda. Email:\url{abdou@aims.edu.gh}.}
\and
Júlia Baligács\footnote{University of Oxford. Email:\url{jbaligacs.gmail.com}}
\and
Max Hahn-Klimroth \footnote{Goethe University Frankfurt, Email:\url{hahnklim@mathematik.uni-frankfurt.de}. Partly supported by DFG FOR 2975.}
\and
Jan Hązła\footnote{AIMS Rwanda. Email:\url{jan.hazla@aims.ac.rw}.
A.M.A.~and J.H.~were
supported by the Alexander von Humboldt Foundation German research chair
funding and the associated DAAD projects No. 57610033
and No. 57761435.}
\and
Lukas Hintze \footnote{Universität Hamburg, Email:\url{lukas.rasmus.hintze@uni-hamburg.de}.}
\and
Olga Scheftelowitsch \footnote{TU Dortmund, Email:\url{olga.scheftelowitsch@udo.edu}. Supported by DFG CO 646/5.}
}

\date{}

\begin{document}

\maketitle

\begin{abstract}
    Polarization and unexpected correlations
    between opinions on diverse topics 
    are an object of sustained attention. However, numerous theoretical models
    do not seem to convincingly explain these phenomena.
    
    This paper is motivated by a recent line of work,
    studying models where polarization can be
    explained in terms of biased assimilation and geometric interplay 
    between opinions on various topics.
    The agent opinions are represented as unit vectors on a multidimensional
    sphere and updated according to geometric rules.
    In contrast to previous work, we focus on the classical opinion exchange
    setting, where the agents update their opinions in discrete time steps,
    with a pair of agents interacting randomly at every step.
    
    Our findings are twofold. First, polarization appears to be ubiquitous
    in the class of models we study, requiring only relatively modest assumptions on the update functions,
    reflecting biased assimilation. Second, there is 
    a qualitative difference between two-dim\-en\-sional
    dynamics on the one hand, and three or more dimensions on the other.
    Accordingly, we prove almost sure polarization for a large class of update rules in two dimensions.
    Then, we prove polarization in three and more dimensions in
    more limited cases and try to shed light on central difficulties
    absent in two dimensions.
\end{abstract}

\section{Introduction}\label{sec:Intro}

It is of great interest in politics and economy to understand how individuals in groups and societies exchange and update their opinions based on influences by new events and mutual interactions.
A salient aspect of opinion evolution is the tendency towards \emph{polarization}.
Numerous dynamics are at play, and the exact definition of polarization is not always apparent. 
For example, \emph{issue alignment} refers to the correlation between opinions of different topics, i.e., if two individuals agree on one topic, they are typically more likely to also agree on some other topic.
Another aspect is \emph{issue radicalization} where individuals build strong opinions on a specific topic, sharply different from other individuals.
In this paper, by \emph{polarization}, we refer to the extreme case of issue alignment where a society splits into two groups where the individuals in the same group mostly agree but mostly disagree with individuals from the other group.
Several authors from a variety of backgrounds are increasingly observing polarization and warning about this phenomenon, e.g., in the context of the United States \cite{BG08,mason2018uncivil,french2020divided,klein2020we}.

More broadly, the desire to understand the evolution
of opinions in societies
has led to extensive theoretical research on modelling and analyzing the relevant dynamics (see, for instance, surveys~\cite{MT17, grabisch2020survey}).
While many of the suggested models, especially those that do not allow for interactions between multiple topics, do not seem to account well for polarization, a recent line of work postulated models where polarization emerges more naturally \cite{HJMR, gaitonde2021polarization}.
Those models are \emph{geometric} in the sense that opinions of individuals (also referred to as an \emph{agents}) are represented as points in the Euclidean space evolving according to geometric rules, in particular based on the scalar
products between interacting opinions.

In particular, Hązła, Jin, Mossel and Ramnarayan~\cite{HJMR} introduced a model where the agent opinions
are represented as vectors on a multidimensional unit sphere. Therein, each dimension may represent a topic or idea (e.g., climate change, gun laws), and the coordinate
value reflects the individual’s stance on the topic. 
The assumption that opinions belong to the unit sphere reflects a fixed ``budget of attention'' of agents, i.e., if an individual is very interested in one specific topic, then they have less time to form a strong and confident opinion about some other topic.
An ``interaction'' updates opinions in a way
that reflects \emph{biased assimilation}:
Agents become ``attracted'' closer to opinions that
align well with their existing beliefs and ``repulsed'' by opposite opinions.
\cite{HJMR} then showed some scenarios leading to polarization, where over time agent opinions
converged to two antipodal groups.

Gaitonde, Kleinberg and Tardos \cite{gaitonde2021polarization}
studied this and related models further.
Pursuing the question of how ubiquitous 
the kind of polarization from~\cite{HJMR} is, 
they established a more general framework
and showed polarization results 
(including ``weak polarization'' in some cases)
in several other scenarios.

However, we believe that the question of ubiquity of polarization
in geometric opinion models is still far from explored.
In this work, we analyze a general class of models 
and show polarization in a wide class of scenarios of ``pure'' opinion exchange.

First, both~\cite{HJMR} and~\cite{gaitonde2021polarization} focus
on models where opinions evolve due to ``external'' influences, in the sense that every individual experiences the same impacts, independent of their initial opinion (representing, e.g., influences by the media).
In contrast, we study the classical
opinion exchange setting, where opinions evolve
due to pairwise interactions between agents, representing, e.g., one-on-one political discussions.

Second, we aim to work with more general update rules in order to understand what properties and assumptions drive polarization.
In~\cite{HJMR}, the updates were proportional 
to the scalar product between two opinions. In particular,
in their model, opinion $\vu$ ``attracts'' in exactly the same
measure as the opposite opinion $-\vu$ ``repulses''. Therefore,
the effect of intervening with opinions $\vu$ and $-\vu$ is exactly the same.
Similarly, most of the examples studied in~\cite{gaitonde2021polarization}
have this property, which they call ``sign-invariance''. 
We consider more general update rules not relying on this exact property.
We think this captures some more realistic,
asymmetric scenarios. For example, 
it seems natural to consider scenarios where opinions 
which are similar to an agent's opinion are likely 
to have a large attraction, while opposing opinions 
elicit only mild dislike.

Our main finding is that issue alignment and polarization
are ubiquitous in the settings that we study. 
We prove almost sure convergence to antipodal pairs of opinions for a robust
class of natural update functions. It appears that
the geometric interplay between multiple topics
and a modest version of biased assimilation
are sufficient to inevitably cause a strong version 
of issue alignment.

On the technical side, our results point towards a qualitative difference in the dynamics when the model is considered in dimensions $\geq 3$ opposed to dimension $2$, i.e., when the model represents more topics.
More precisely, let $d$ denote the dimension of the model in the sense that opinions are vectors on the $(d-1)$-dimensional sphere $\mathbb{S}^{d-1}\subseteq \mathbb{R}^d$.
For $d=2$, we show polarization for a natural class of \emph{stable} ({see \Cref{def:stable_function}) update functions unless there is a trivial obstacle preventing it (\Cref{thm:2d-main}).
However, in \Cref{sec:counterexample}, we show that an important property that we rely on in the proof of the two dimensional case does not hold for a larger number of dimensions.

For $d\ge 3$, we show almost sure polarization in
case of three agents (\Cref{thm:three-ops}). We also show almost sure
polarization for another natural class of \emph{active} update functions
(\Cref{thm:active}), for which the polarization is somewhat easier to establish.
Finally, we attempt to isolate the challenge in the case of more than three
opinions and stable functions. We define a notion of \emph{$(d,n)$-stable} 
update functions, which captures the crucial needed property of
escaping ``almost orthogonal'' configurations.
We prove that such functions
induce almost sure polarization.
However, we leave for future work to show that any specific function
actually has this property.

\paragraph{Related work}\label{par:related_work}
As mentioned, the two works which are most relevant to this one
are~\cite{HJMR} and~\cite{gaitonde2021polarization}. We refer an interested 
reader to these papers for additional literature discussion.

These two works mostly study settings where the opinions evolve due to
external influences. One could surmise that the polarization arises due to 
those external ``interventions''.
Here, we argue that it is inevitable also when such influences are absent and agents only interact with each other. 
Therefore, the polarization is more likely to be a general feature of the geometric structure utilized in this class of models.
Furthermore, we believe that there is value in 
rigorously showing that polarization arises generally rather than 
as a special property of specific update functions.

On a technical level, it might appear intuitive that
polarization occurs in our setting of pure opinion exchange.
However, we face challenges not present in models discussed above
due to lack of additional ``mixing'' introduced by external influences.
In particular, as discussed throughout the paper,
for many natural update functions
it seems challenging to deal with configurations where the agents are almost, but not entirely, partitioned into two orthogonal groups.
Furthermore, in some of the proofs we rely on potential functions to a smaller extent. Therefore, our proofs use some novel arguments compared
to~\cite{HJMR} and~\cite{gaitonde2021polarization}.

The ``party model'' from~\cite{gaitonde2021polarization} 
(see equation~(4) therein)
is arguably the  most similar to our setting.
In short, it is a mixture of external influences and (averaged) pairwise interactions. We think the closest comparison is to our model with the sign update function. In contrast, we study a different setting with only pairwise interactions and more general class of update functions. In particular, we believe the "stable" functions from our paper are technically more challenging. (The "party model" is more general in that it only requires a connected graph of interactions while our results are for fully supported distributions.)

Many classical opinion exchange models have been studied that tend to induce consensus
rather than disagreement.
One of the most striking examples is Aumann's agreement theorem for Bayesian models \cite{Aum76},
see also~\cite{GK03, MT17} for Bayesian models more generally. 
Further well-known examples include models introduced by DeGroot~\cite{DeGroot}, and Friedkin
and Johnsen~\cite{friedkin1990social}.
On the other hand, there is vast literature addressing various aspects of polarization.
In contrast to our work, much of this literature is focused on one-dimensional 
opinions
(e.g.,~\cite{MKFB03, DGL13, DVSC17, yang2022hybrid, chen2022polarizing}),
opinions or actions with discrete rather than continuous values (e.g.,~\cite{Axe97, PY18})
or update rules quite different from our geometric framework
(e.g.,~\cite{BB07,mi2023polarization}).

\section{Model and dynamics}\label{sec:model_and_dynamics}
In our model, we have a total of $n$ agents, each having an opinion represented as a vector on the unit sphere \(\mathbb{S}^{d-1}\). We denote the opinion of the $i$-th agent by $\vec{u}_i$. As explained above, each dimension may represent a topic, and each coordinate of $\vec{u}_i$ reflects the individual's stance on a given topic.

We then consider the dynamics with discrete time $t=0,1,2,\ldots$, where at each step $t$, a pair $(i,j)$ is chosen i.i.d.~from $[n]\times[n]$ according to a fixed distribution $\mathcal{D}$ with full support, referred to as an \emph{interaction distribution}.
We then say that, at time $t$, agent $j$ \emph{influences} agent $i$, or more generally
that $i$ and $j$ \emph{interacted} at time $t$. 
Then agent $i$ updates its opinion, where we define its new opinion \(\vec u_i'\) to be given by 
\begin{equation}
    \vec u_i' = \frac{\vec w}{\norm{\vec w}} , \; \text{where} \;
    \vec w = \vec u_i \, + \, f\left(\bck{\vec u_i, \vec u_j}\right) \cdot \vec u_j \;
    \label{eq:06}
\end{equation}
where $f:[-1,1]\to\mathbb{R}$ is a fixed \emph{update function} (we prove in \Cref{claim:new_correlation} that, in the update functions we consider, we never have $\vec w=0$. Therefore, the new opinion is well defined.).

We refer to a tuple of all agents' opinions  $\mathcal{U}=(\vu_1,\ldots,\vu_n)\in (\mathbb{S}^{d-1})^n$ as a \emph{configuration} and let $\mathcal{U}^{(t)}$ denote the configuration at time $t$. 
The most extreme version of issue alignment is captured by the following definition.
\begin{definition}\label{def:polarization}
A configuration $\mathcal{U}$ is \emph{polarized} if for every
$i,j$ either $\vu_i=\vu_j$ or $\vu_i=-\vu_j$.

We say that a sequence of configurations $\mathcal{U}^{(t)}$ polarizes
if  $\ \lim_{t\to\infty}\mathcal{U}^{(t)}$ exists and is a polarized configuration
(where convergence is in the standard topology
in $\mathbb{R}^{d}$).
\end{definition}
Note that, according to this definition, a consensus configuration is also considered polarized. 
While this might initially seem undesirable, it is unavoidable if we want to prove polarization for arbitrary initial configurations: Indeed, if all opinions are sufficiently close in the starting configuration, they will converge to consensus.
However, we show in \Cref{lem:concentration} that, when the initial configuration is drawn from a distribution satisfying certain symmetry conditions, the agents with high probability form two antipodal groups of roughly equal size. In other words, 
``true'' polarization occurs with high probability.

In the following, we study the conditions under which
a sequence of configurations $\mathcal{U}^{(t)}$ polarizes almost surely.

\subsection{Update functions}\label{subsec:update_functions}
The update function
signifies how much agent $i$ will change their opinion based
on the preexisting alignment $\bck{\vu_i,\vu_j}$. 
The biased assimilation assumption should reflect that
    if the opinions are similar (e.g., $\bck{\vu_i, \vu_j} \approx 1$), $\vu_i$ will be ``attracted'' towards
    $\vu_j$,
    and if they are dissimilar (e.g., $\bck{\vu_i, \vu_j} \approx -1$), then $\vu_i$ is ``repelled'' away
    from $\vu_j$.
Some natural examples of update functions are
$f(A)=\eta\cdot A$ suggested in~\cite{HJMR}
and $f(A)=\eta\cdot \sign(A)$ studied in~\cite{gaitonde2021polarization}
(in both cases there is a hyperparameter
$\eta>0$).

However, there are many
other natural possibilities. For example, 
we can assume different strengths of ``positive''
and ``negative'' interactions
i.e., for
some \(\eta_{+},\eta_{-}>0\),
\begin{equation}
\label{eq:asymmetric}
    f(A)=\begin{cases}
    \eta_+\cdot A&\text{ if $A\ge 0$,}\\
    \eta_-\cdot A&\text{ if $A< 0$.}
\end{cases}
\end{equation}
Another alternative is the asymmetric sign function
$f(A)=\eta_+\cdot \mathbb{I}[A\ge 0]-\eta_-\cdot\mathbb{I}[A<0]$
(or even $f(A)=\eta_+\cdot\mathbb{I}[A\ge t]
-\eta_-\cdot\mathbb{I}[A<t]$ for some $-1<t<1$).

\paragraph{Moving directly on the sphere}\label{par:moving_directly_on_the_sphere}
Our framework also captures update rules that can be expressed 
differently, arguably in a more natural way.
For example, consider the following scheme that avoids explicit normalization
from~\eqref{eq:06}:

Let $0<\eta<1$
and $\alpha = \cos^{-1}\left(\bck{\vu_i, \vu_j}\right)$ be the angle between $\vu_i$ and $\vu_j$ (i.e., the distance on the sphere between $i$ and $j$'s opinions).
Let $\phi = \frac{\eta}{2} \sin(2\alpha)$ be the distance by which $i$ will move its opinion towards/away from $j$ (positive sign: towards; negative sign: away from; other choices
are possible but a nice property of this one is $\phi\ge 0$ iff
$\alpha\le \pi/2$ iff $\langle\vu_i,\vu_j\rangle\ge 0$).

Then, set 
\begin{equation}\label{eq:07}
    \vu_i' = (\sin(\alpha - \phi) \vu_i + \sin(\phi) \vu_j) / \sin(\alpha)\;.
\end{equation}
The formula~\eqref{eq:07}
is known as ``spherical linear interpolation'' \cite{Shoemake85}
and the new opinion $\vu_i'$ is guaranteed to be on the sphere. 
It can be easily checked that the procedure we just described gives the same updates
as~\eqref{eq:06} with the update function 
$f_{\slerp,\eta}(A)=\sin\phi/\sin(\alpha-\phi)$.

\paragraph{Stable and active functions}\label{par:stable_and_active_functions}
Our aim is to cover a broad class of update functions
while keeping the mathematical discussion tractable.
To that end, let us settle on two definitions:

\begin{definition}[Stable function]
\label{def:stable_function}
    A function \(f : [-1, 1] \to \mathbb{R}\) is \emph{stable} if
    it is continuous and if 
    $\sign{(f(A))} = \sign{(A)}$ for all $A \in [-1,1]$.
\end{definition}

Note that every stable function $f$ fulfills $f(0)=0$ and $f(A)$ is positive (respectively negative) if and only if $A$ is positive (respectively negative).
The scaled identity  
$f(A)=\eta A$,
as well as its asymmetrical version from~\eqref{eq:asymmetric}
and $f_{\slerp,\eta}$ are all stable.

\begin{definition}[Active function]
    \label{def:active_function}
    An update function $f$ is \emph{active} if
    \begin{enumerate}
        \item There exists $-1<A_0<1$ such that
        $f(A)\le 0$ for $A< A_0$ and $f(A)\ge 0$ for $A>A_0$.
        \item There is a universal lower bound $m>0$
        such that $\inf_{-1\le A\le 1}|f(A)|\ge m$.
    \end{enumerate}
\end{definition}

The sign function $f(A)=\eta\cdot(\mathbb{I}[A\ge 0]-\mathbb{I}[A<0])$ (as well as its asymmetrical variants
discussed above) is an example of an active function.

Our definitions include a requirement of at most one sign change
of $f$ (located exactly at $f(0)=0$ in case of a stable function).
This does not seem to be fundamental, but we impose these
conditions for the sake of simplicity (among others, in order
to keep the concept of ``active configurations'' simpler).
Nevertheless, we believe
they capture the most natural and mathematically interesting
cases.

\subsection{Examples}
\label{sec:examples}
While \Cref{def:polarization}
reflects complete issue alignment, it does not necessarily
capture the division into two sizable groups with opposing opinions. In particular, 
it includes consensus configurations where all opinions are equal.

If the update function is odd $f(-A)=-f(A)$, then there is a symmetry argument
(observed previously both in~\cite{HJMR} and~\cite{gaitonde2021polarization})
to establish the division into two opposing camps. 
More precisely, if $n$ opinions are sampled i.i.d.~from
a distribution which gives the same weights/density to $\vu$ and $-\vu$, then
with high probability the agents converge to a polarized configuration
with $n/2\pm O\left(\sqrt{n}\right)$ agents in each group.
For completeness we give a concentration bound
argument in \Cref{lem:concentration}.

For asymmetric update functions, we do not address this issue theoretically.
However, empirically it appears that even if the update function is asymmetric,
randomly placed agents tend to congregate into two groups, each of significant size.

\begin{figure}
   \caption{Issue radicalization with an asymmetric update rule: \(n = 100, d = 4, \eta_+ = 0.9, \eta_- = 0.1\). The initial configuration is obtained by choosing \(n\) opinions uniformly at random on the unit hyper-sphere. The update rule used is the one in \eqref{eq:asymmetric}.
    The first image shows what such initial configuration looks like. Each of \Cref{fig:asym_b,fig:asym_c,fig:asym_e,fig:asym_f} shows the configuration at time \(t = 2800\), for four independent runs of the experiments (i.e., four different initial configurations).
    Furthermore, over \(1000\) independent runs of the experiments, the histogram of the size of the polarized groups is plotted.
    As $d=4$, three coordinates are plotted
    spatially, and
    color represents the fourth coordinate.}
    \label{fig:asymmetric}
    \centering
    \begin{subfigure}[c]{0.3\textwidth}
        \centering
        \caption{\footnotesize A uniform configuration}
        \includegraphics[width=5cm]{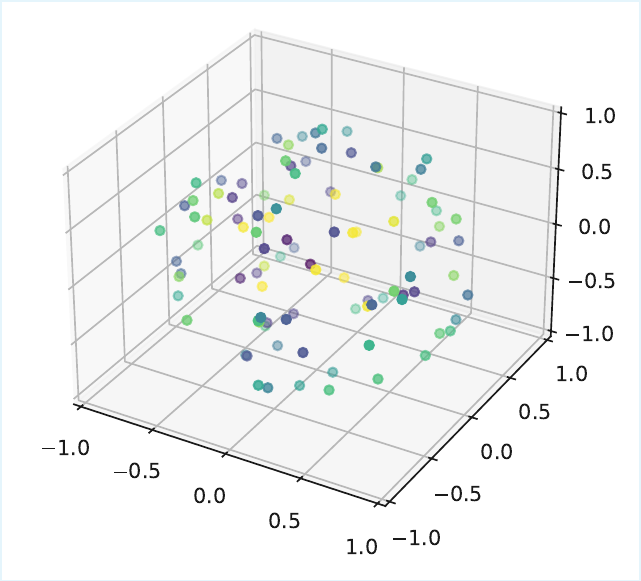}
        \label{fig:asym_a}
    \end{subfigure}
    \begin{subfigure}[c]{0.3\textwidth}
        \centering
        \caption{\footnotesize First run: $t=2800$}
        \includegraphics[width=5cm]{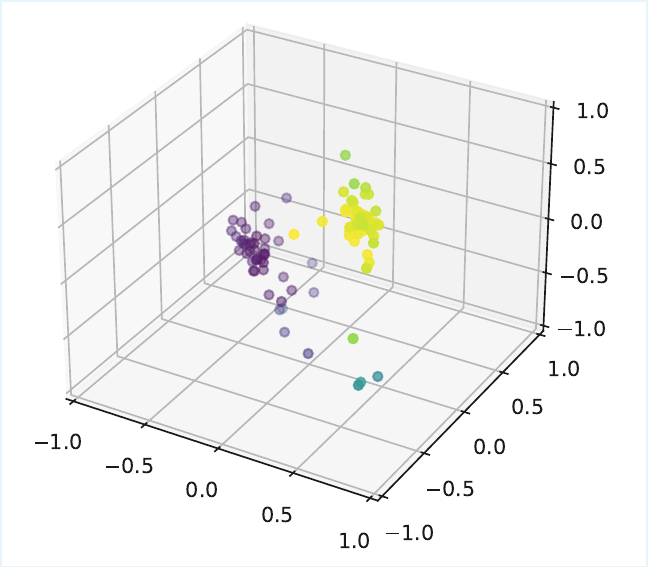}
        \label{fig:asym_b}
    \end{subfigure}
    \begin{subfigure}[c]{0.3\textwidth}
        \centering
        \caption{\footnotesize Second run: $t=2800$}
        \includegraphics[width=5cm]{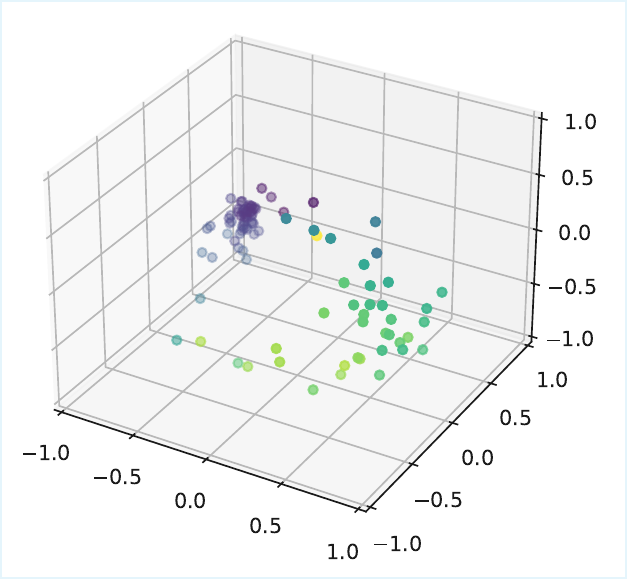}
        \label{fig:asym_c}
    \end{subfigure}
    \\
    \begin{subfigure}[c]{0.35\textwidth}
        \centering
        \caption{\footnotesize Histogram (size of clusters)}
        \includegraphics[height=4.5cm]{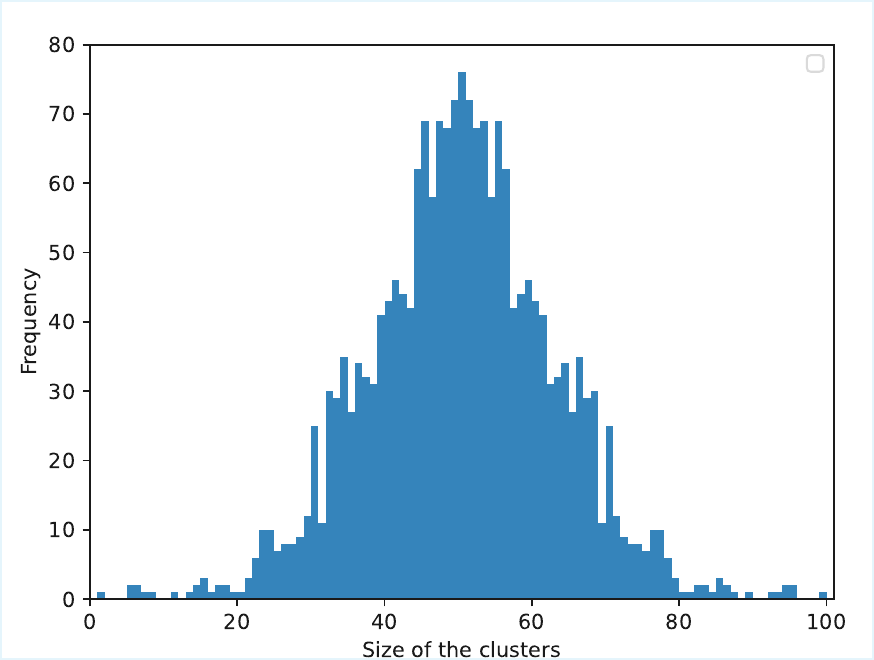}
        \label{fig:asym_d}
    \end{subfigure}
    \begin{subfigure}[c]{0.3\textwidth}
        \centering
        \caption{\footnotesize Third run: $t=2800$}
        \includegraphics[width=5cm]{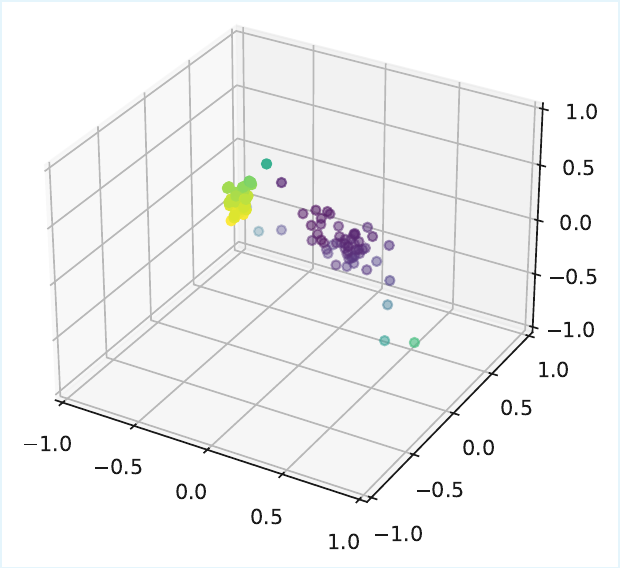}
        \label{fig:asym_e}
    \end{subfigure}
    \begin{subfigure}[c]{0.3\textwidth}
        \centering
        \caption{\footnotesize Fourth run: $t=2800$}
        \includegraphics[width=5cm]{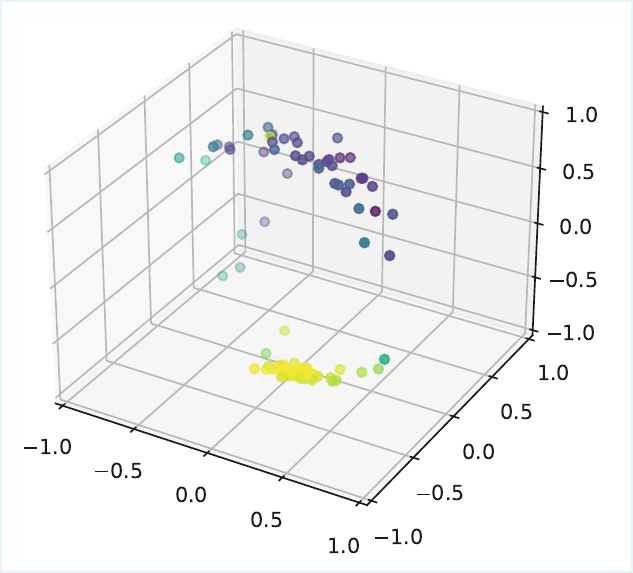}
        \label{fig:asym_f}
    \end{subfigure}
\end{figure}

For example, \Cref{fig:asymmetric} illustrates simulation 
for the update function given by~\eqref{eq:asymmetric} with~\(\eta_+ = 0.9\), and \(\eta_- = 0.1\).
Opinions of \(100\) agents were initially placed uniformly at random on the sphere for
$d=4$. An example of such random configuration is shown in \Cref{fig:asym_a}.
For four different random starting configurations, the configurations after \(2800\) time steps are shown in the image. The histogram on \Cref{fig:asym_d} displays the distribution of the size of the polarized groups, over \(1000\) runs of the experiment. We can see that despite the asymmetry of the update rule, the group size is around half of the number of opinions.

The most challenging aspect of the theoretical analysis we employ is to determine if 
the random process can get stuck in an ``almost orthogonal'' configuration.
That could correspond to a situation where the agents sort into communities
interested in disjoint topics. 
However, such almost orthogonal configurations are inherently unstable and 
ultimately evolve into polarization.
\begin{figure}
    \caption{Almost orthogonal configurations polarize: \(n = 100, d = 4\) and the update rule is \(f(A) = 0.1 A\). Initially, for every pair of opinions, their dot product \(A\) satisfies \(\abs{A} \leq 0.001\) or \(\abs{A} \geq 0.999\). The process escapes the ``almost orthogonal'' configuration and polarizes.}
    \label{fig:almost_orthogonal}
    \centering
        \begin{subfigure}[c]{0.3\textwidth}
        \centering
        \caption{\footnotesize $t=0$}
        \includegraphics[width=5cm]{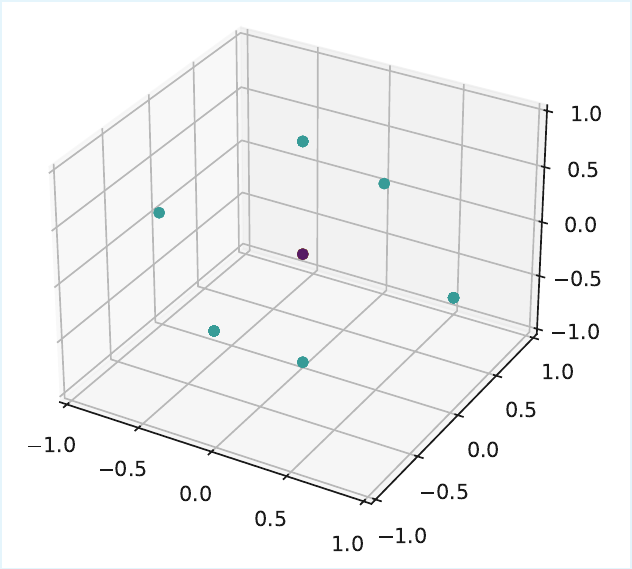}
        \label{fig:orthogonal_a}
    \end{subfigure}
    \centering
        \begin{subfigure}[c]{0.3\textwidth}
        \centering
        \caption{\footnotesize $t=17000$}
        \includegraphics[width=5cm]{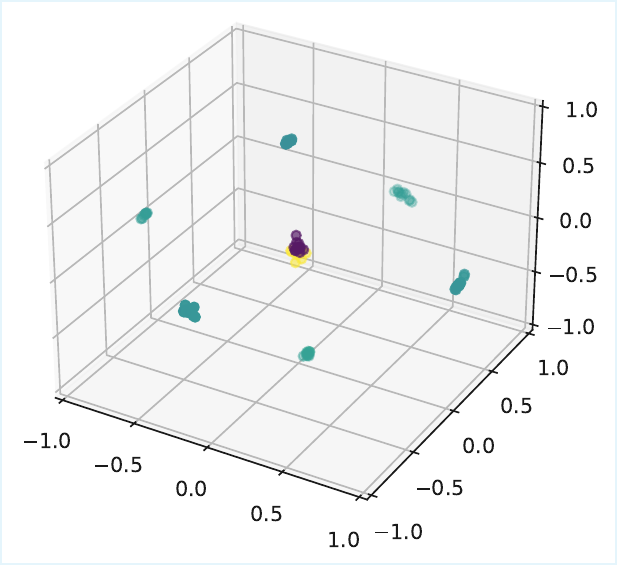}
        \label{fig:orthogonal_b}
    \end{subfigure}
    \centering
        \begin{subfigure}[c]{0.3\textwidth}
        \centering
        \caption{\footnotesize $t=19000$}
        \includegraphics[width=5cm]{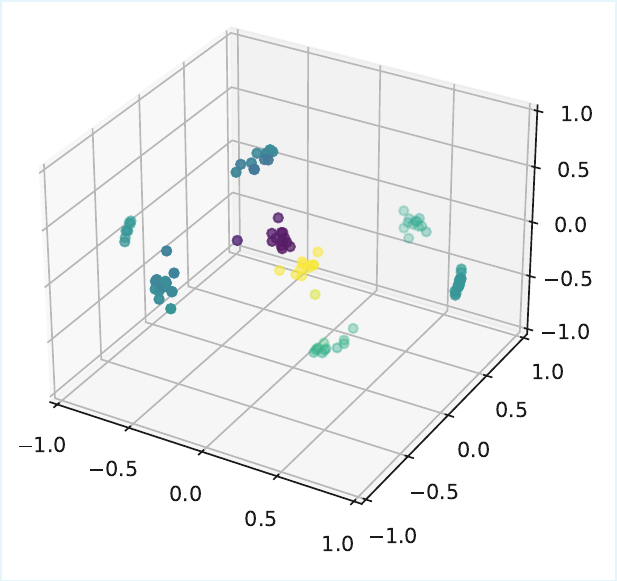}
        \label{fig:orthogonal_c}
    \end{subfigure}
    \\
    \centering
        \begin{subfigure}[c]{0.3\textwidth}
        \centering
        \caption{\footnotesize $t=20000$}
        \includegraphics[width=5cm]{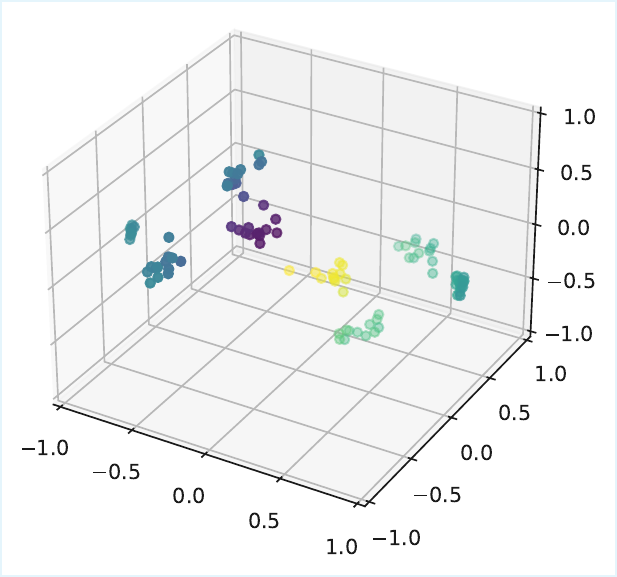}
        \label{fig:orthogonal_d}
    \end{subfigure}
    \centering
        \begin{subfigure}[c]{0.3\textwidth}
        \centering
        \caption{\footnotesize $t=21000$}
        \includegraphics[width=5cm]{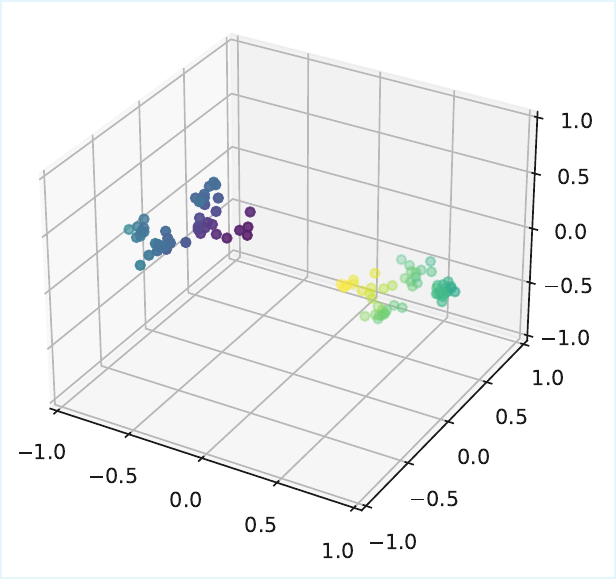}
        \label{fig:orthogonal_e}
    \end{subfigure}
    \centering
        \begin{subfigure}[c]{0.3\textwidth}
        \centering
        \caption{\footnotesize $t=23000$}
        \includegraphics[width=5cm]{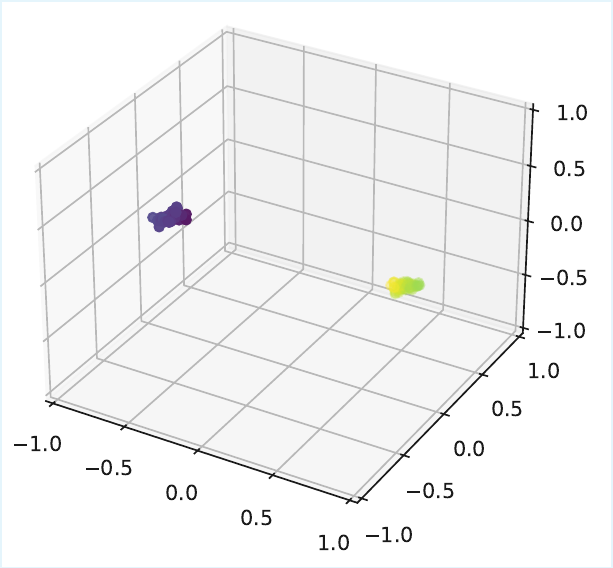}
        \label{fig:orthogonal_f}
    \end{subfigure}
\end{figure}
\Cref{fig:almost_orthogonal} illustrates simulations with the update rule \(f(A) = \eta A\) with \(\eta = 0.1\). We took \(100\) agents in \(4\) dimensions, and the starting configuration is chosen such that, for every two opinions \(\vu_i, \vu_j\), either \(\abs{\bck{\vu_i, \vu_j}} \leq 0.001\) (i.e. ``almost orthogonal'') or \(\abs{\bck{\vu_i, \vu_j}} \geq 0.999\). 
The configurations at different time steps show that the process escapes the almost orthogonal configuration, and eventually polarizes.

\subsection{Results}\label{subsec:results}
Our results address the following question: In which settings does a given
starting configuration $\mathcal{U}^{(0)}$ polarize almost surely?

First, we show polarization in the case of active update functions.
\begin{theorem}
\label{thm:active}
Let $d,n\ge 2$ and $f$ be an active update function,
$\mathcal{D}$ a fully supported interaction distribution
and $\mathcal{U}^{(0)}$ a configuration of $n$ opinions.
Then, $\mathcal{U}^{(t)}$ polarizes almost surely.
\end{theorem}

It turns out that the case of stable update functions is significantly more challenging.
In their case, we first need to identify special cases that do not polarize for trivial reasons.

\begin{definition}[Separable configuration]
\label{def:separable}
    A configuration \(\mathcal{U}\) is separable when it can be partitioned into two nonempty sets \(S\) and \(T\) such that,
    for every opinion \(\vu \in S\) and \(\vec v \in T\), we have \(\vu \perp \vec v\).
\end{definition}

Since stable update rules satisfy $f(0)=0$, a separable configuration remains separable with the same partition after any interaction. In particular, such a configuration cannot polarize.
For this reason, when considering stable update functions, we always assume that the initial configuration is non-separable.
Conversely, in \Cref{lem:non-separable-preserved} we show that non-separability is preserved under any interaction. Accordingly, we ask
the question:
Under which conditions does the process
$\cU^{(t)}$ polarize almost surely for all non-separable initial configurations? 
In particular, does it polarize almost surely
for all stable update functions?

We provide partial positive answers to this
question by showing polarization for stable update functions in two special cases.
First, we obtain polarization in the case of three agents.

\begin{theorem}
\label{thm:three-ops}
    Let $d\ge 2$, $f$ be a stable update function, the interaction distribution
    $\mathcal{D}$ have full support and $\mathcal{U}^{(0)}$ be a configuration
    of three opinions which is not separable.
    Then, $\mathcal{U}^{(t)}$ polarizes almost surely.
\end{theorem}

Second, we show polarization in the two-dimensional setting for an arbitrary number of opinions.

\begin{theorem}
\label{thm:2d-main}
Let $d=2$, \(n\geq2\),
the interaction distribution $\mathcal{D}$ have full support,
$f$ be a stable update function,
and $\mathcal{U}^{(0)}$ be any initial configuration which is not separable. 
Then, almost surely, $\mathcal{U}^{(t)}$ polarizes.
\end{theorem}

\begin{remark}
The only aspect of the distribution $\mathcal{D}$
that plays a role in the proofs is the minimum probability $p_{\mathrm{min}}=\min_{i,j}\mathcal{D}(i,j)$. Our theorems assume that $\mathcal{D}$
is fully supported, equivalently $p_{\mathrm{min}}>0$.
It seems natural and interesting future direction to extend our study
to more general distributions.

For convenience, we assume throughout that $\mathcal{D}$
is fully supported on the set $[n]\times[n]$.
In fact, since an interaction $(i, i)$ does
not change the configuration,
for the theorems above
$\mathcal{D}$ needs to be fully supported 
only on the set $\{(i,j)\in[n]\times[n]: i\ne j\}$.
\end{remark}

\subsection{Outline of proofs}
\label{sec:proof-outline}

Let us briefly discuss the main ideas used to prove our results.

\paragraph{Strictly convex configurations}
For all three theorems stated above, 
we found it useful to 
consider a class of configurations that
we call \emph{strictly convex}:
\begin{definition}[Strictly convex configuration]
    \label{def:strictly_convex}
    Let $0 \leq c < 1$. A configuration $\mathcal{U}$ is called \emph{$c$-strictly convex} if there
    exist $b_1,\ldots,b_n\in\{\pm 1\}$ such that 
    $\bck{b_i\vu_i,b_j\vu_j}>c$ for every $1\le i,j\le n$.
    When \(c = 0\), we simply say that \(\mathcal{U}\) is strictly convex.
\end{definition}

Equivalently, a configuration is strictly convex if there exists a choice
of $b_1,\ldots,b_n$ and a convex cone with apex angle less than $\pi/4$  that contains
every $b_i\vu_i$. The apex angle of a cone is the angle between its axis and any of its generating vectors. If an update function
satisfies $\sign(f(A))=\sign(A)$, it is not hard to see that for any interaction of agents,
the updated configuration remains strictly convex
with the same $b_1,\ldots,b_n$. And, since a new opinion of an agent is a linear combination
of two existing opinions, the ``width'' of the cone of the updated configuration cannot increase. More precisely,
if we let 
\begin{equation*}
\mathcal{P}(t)=1-\min_{i,j\in[n]}\bck{b_i\vu_i^{(t)},b_j\vu_j^{(t)}}\;,
\end{equation*}
then $\mathcal{P}(t+1)\le\mathcal{P}(t)$ always holds.
Furthermore, it can be shown that with probability
uniformly bounded away from zero, $\mathcal{P}(t)$
decreases by a constant factor in a constant number of steps,
hence $\lim_{t\to\infty}\mathcal{P}(t)=0$ happens almost surely.
As a result, the configuration polarizes,
where the opinions of agents with $b_i=1$ converge
to some vector $\vu^*$, and the
opinions of agents with $b_i=-1$
to the antipodal vector $-\vu^*$.
Ultimately, we show:

\begin{lemma}
\label{lem:strictly-convex-intro}
Let $d,n\ge 2$ and
$\mathcal{D}$ be a fully supported
distribution.

If the update function is stable and $\cU^{(0)}$ strictly
convex, then $\cU^{(t)}$ polarizes almost surely.

If the update function is active, recall
$A_0$ from \Cref{def:active_function}.
If $\cU^{(0)}$ is $|A_0|$-strictly convex,
then $\cU^{(t)}$ polarizes almost surely.
\end{lemma}

\Cref{lem:strictly-convex-intro} is proved
in \Cref{sec:strictly-convex}. The difference
in the case of active update functions stems from
the fact that they do not satisfy $\sign(f(A))=\sign(A)$
in general, but only for $|A|> |A_0|$.

\paragraph{Proof outlines of \Cref{thm:active} and \Cref{thm:2d-main}}

The proofs of both theorems follow a common strategy.
In both settings, active update functions for $d\ge 2$ and stable update functions for $d=2$, we establish that
a strictly convex configuration can be
reached in a constant number of steps from any allowed
initial configuration.

\begin{lemma}
\label{lem:almost_surely_c_strictly_convex}
    Let $d,n\ge 2$, and $f$ an active update function.
    There exists $K$ such that, for
    every configuration $\mathcal{U}^{(0)}$, there is a sequence
    of $K$ interactions after which $\mathcal{U}^{(K)}$ is
    $|A_0|$-strictly convex.
\end{lemma}

\begin{lemma}
    \label{lem:opinions_in_a_quadrant}
    Let $d=2$, $n\ge 2$, and $f$
    a stable update function.
    There exists~\(K\) such that, for every configuration
    \(\mathcal{U}^{(0)}\) which is
    not separable,
    there exists a sequence of $K$ interactions such that
    $\mathcal{U}^{(K)}$ is strictly convex.
\end{lemma}

Note that the constant $K$ in the lemmas above
can depend on $d,n$, and $f$. This is true
also for other similar statements in the rest of the paper.

Since any sequence of
interactions of constant length occurs with some positive constant probability, 
at every time there is a constant probability
that a strictly convex configuration will
be reached in $O(1)$ steps.
It follows that a strictly convex configuration
will be reached almost surely.
Then, polarization follows from  \Cref{lem:strictly-convex-intro}.

\Cref{lem:almost_surely_c_strictly_convex} is easier to handle,
since for active update functions the uniform lower bound $|f(A)|\ge m$ ensures
that any pair of opinions will have their
absolute correlation exceed $|A_0|$ (or any other
fixed value) after $O(1)$ interactions between them.
Therefore, the required sequence of interactions
is constructed by agent $1$ influencing each
other agent a fixed number of times.
The details are given in \Cref{sec:active}.

\Cref{lem:opinions_in_a_quadrant} seems more surprising, and its proof is more involved.
Furthermore, as discussed shortly, its statement
is false for $d\ge 3$. To understand this, it is useful
to distinguish configurations which we call \emph{$\eps$-inactive}.

\begin{definition}[\(\epsilon\)-activity]
\label{def:epsilon_activity}
    Let \(\epsilon\ge 0\). A configuration \(\mathcal{U}\) is called \(\epsilon\)-inactive if for every $i,j$ it holds
$\abs{\bck{\vu_i,\vu_j}}\le \eps$ or
$\abs{\bck{\vu_i,\vu_j}}\ge 1-\eps$.
    Otherwise, we say that the configuration is $\eps$-active.
    When \(\epsilon = 0\), we simply say active
    and inactive instead of \(0\)-active (respectively \(0\)-inactive).
\end{definition}

For a stable update function, if
$A_{ij}=\bck{\vu_i,\vu_j}\approx 0$, then
also $f(A_{ij})\approx 0$.
In that case, 
it can be easily seen from~\eqref{eq:06} that
an interaction
where $j$ influences $i$ will not move the
opinion of agent $i$ by much.
On the other hand, if
$|A_{ij}|\approx 1$, then,
since $\vu_i$ and $\vu_j$
(or $\vu_i$ and $-\vu_j$ if their correlation is
negative) are already close to
each other, and since
the updated opinion lies on the great circle
between them,
the updated opinion will not
move far away from the original one.
Therefore, every interaction in an
$\eps$-inactive configuration effects
only a slight change.

Turning back to \Cref{lem:opinions_in_a_quadrant}, 
if a non-separable initial configuration is also
$\eps$-inactive,
it seems less obvious to construct the 
required sequence of interactions of bounded length. 
In \Cref{sec:2d} we give a geometric argument to achieve that.

\paragraph{Polarization of three opinions}
Finally, \Cref{thm:three-ops}
proved in \Cref{sec:stable-preliminaries}
gives almost sure polarization for stable functions and three agents. 
Here, we introduce another approach.
We define a suitable potential function and show that it decreases almost surely to some constant value.
Importantly, the speed of decrease depends on the starting configuration.
Then we prove that, once this constant value is reached, there exists a constant-length sequence making the configuration strictly convex, so that
\Cref{lem:strictly-convex-intro} can be applied.

\subsection{Partial progress for stable functions}\label{subsec:partial_progress_for_stable_functions}

A natural follow-up question is: Which stable update functions elicit almost sure polarization for $d\ge 3$?
Unfortunately, the proof strategy for $d=2$ cannot
be employed directly.
In particular, in general the statement of \Cref{lem:opinions_in_a_quadrant} does not hold.

\begin{theorem}
\label{thm:counterexample}
    Let $d\geq 3$ and consider the update function 
    $f(\langle \vu,\vec v\rangle):=\eta \cdot \langle \vu,\vec v\rangle$ for some $\eta > 0$.
    For every $T\in\mathbb{Z}_{>0}$, there exists a configuration of three opinions
    $\mathcal{U}^{(0)}=(\vu_1^{(0)},\vu_2^{(0)},\vu_3^{(0)})$ such that
    $\mathcal{U}^{(0)}$ is not separable and, for every choice of
    $T$ interactions, the resulting configuration $\mathcal{U}^{(T)}$
    is \emph{not} strictly convex.
\end{theorem}

\Cref{thm:counterexample} is proved in
\Cref{sec:counterexample}. An important point
for our purposes is that the configurations constructed
in the proof of \Cref{thm:counterexample} satisfy
$A_{12}=A_{13}=-A_{23}=\eps$ for
some $\eps=\eps(T)\to 0$ as $T\to\infty$.
In particular, such configurations are
$\eps$-inactive. Clearly, in order to prove almost sure
polarization for~$d\ge 3$, one at least needs to show
that such configurations are almost surely ``escaped''.
This leads us to the following more general considerations.

For small
enough $\eps$, there is a natural structure to
\(\epsilon\)-inactive configurations. 
We can uniquely divide the opinions
into ``clusters'' such that all absolute correlations inside a given cluster are close to 1, and all opinions in different
clusters are almost orthogonal:

\begin{definition}[Cluster]
\label{def:cluster}
    Let \(\mathcal{U}\) be a configuration.
    A non empty set \(C \subset [n]\) is a cluster of \(\mathcal{U}\) if,
    for every \(i, j \in C\), \(\abs{A_{ij}} > 1/2\), and
    for every \(i \in C\), \(j \not\in C\), \(\abs{A_{ij}}<1/2\).
\end{definition}

Of course, in an $\eps$-inactive configuration
for $\eps<1/2$,
the correlations inside a cluster satisfy
$|A_{ij}|\ge 1-\eps$, and correlations between clusters
satisfy $|A_{ij}|\le \eps$. Furthermore, as suggested
above, the clusters of an $\eps$-inactive configuration form a partition of agents.

\begin{lemma}
\label{lem:cluster}
    Let \(0 \leq \epsilon < 1/8\) and \(\mathcal{U}\) a configuration that is \(\epsilon\)-inactive.
    Then, the clusters of~$\mathcal{U}$ form a partition
    of the agents. Furthermore,
    if $\eps<1/d$, then there are at most $d$ clusters.
\end{lemma}

\begin{lemma}
\label{lem:clusters_preserved}
    Let $f$ be a stable update function. 
    There exists $\eps_0> 0$ such that:
    \begin{itemize}
    \item If configuration $\mathcal{U}^{(t)}$ is 
    $\eps_0$-inactive, then $\mathcal{U}^{(t+1)}$ is
    $1/256$-inactive.
    \item Configurations
    \(\cU^{(t)}\) and \(\cU^{(t+1)}\) have the same clusters.
    \end{itemize}
\end{lemma}

\Cref{lem:cluster} is proved in \Cref{sec:preliminaries}
and \Cref{lem:clusters_preserved} in \Cref{sec:3d-stable}.

Clearly, if a sequence of configurations polarizes,
then for large enough $t$
it holds $\left|A_{ij}^{(t)}\right|>1/2$ for every $i,j$, in particular the configuration consists of one cluster
containing all agents.
On the other hand, in light of \Cref{lem:clusters_preserved},
for small $\eps>0$, the number of clusters of an $\eps$-inactive configuration does not change over time
as long as the configuration remains $\eps$-inactive.
Therefore, in order to prove almost sure polarization
for all non-separable initial configurations,
it is necessary to prove in particular that any
non-separable $\eps$-inactive configuration with at least two clusters almost surely eventually becomes $\eps$-active.

In \Cref{sec:3d-stable}, we make partial progress
towards analyzing polarization for stable functions and
$d\ge 3$. We prove that (with one additional caveat), the
property described above is not only necessary, but also
sufficient for almost sure polarization. In order to state
this result,
we introduce a concept of 
\emph{$(d,n,\mathcal{D})$-stable update functions} that captures the property required for polarization:

\begin{definition}
\label{def:dn-stable}
Let $d,n\ge 2$, $\mathcal{D}$ a fully supported distribution, and $f$ a stable update function.
Let $\eps_0$ be given by \Cref{lem:clusters_preserved}.
$f$ is called \emph{$(d,n,\mathcal{D})$-stable} if there exists some~$0<\eps\le\eps_0$ and $p>0$ such that:

Let $\mathcal{U}^{(0)}$ be a configuration of $n$ opinions in
$\mathbb{S}^{d-1}$ which is $\eps$-inactive and not separable, with $k\ge 2$
clusters.
Then, the following two properties hold.
First, almost surely, there exists $t$ such that~$\mathcal{U}^{(t)}$
is $\eps$-active. Second, for the earliest such time $t_1$,
with probability at least $p$,
there exist two agents $i$, $j$ such that at time \(0\) they are in different
clusters, and $\left|A_{ij}^{(t_1)}\right|>\eps$.
\end{definition}

\begin{theorem}
\label{thm:dn-stable}
Let $d,n\ge 2$, $\mathcal{D}$ a distribution with full support,
and $f$ a $(d,n,\mathcal{D})$-stable update function. 
Let $\mathcal{U}^{(0)}$ be a configuration which is not separable.
Then, $\mathcal{U}^{(t)}$ almost surely polarizes.
\end{theorem}

Note the above-mentioned caveat: The definition has an additional condition of $\left|A_{ij}^{(t_1)}\right|>\eps$,
for $i,j$ from different clusters with probability at least $p$. We refer to \Cref{sec:3d-stable} for the discussion
of this condition and the proof of \Cref{thm:dn-stable}.

Since the notion of $(d,n,\mathcal{D})$-stability is given in terms of
behavior of the random process over a potentially long time,
it is not clear how to show that any given stable update function satisfies
the definition.
Indeed, in this paper we do not prove that any specific function
is $(d,n,\mathcal{D})$-stable
for all $d,n\ge 3$, and leave this problem for future work.

\paragraph{Subsequent work}
A follow-up paper \cite{ABH25} shows (with some effort)
that $f(x)=\eta\cdot x$ is $(d,n,\mathcal{U})$-stable for
$\eta>0$, $d,n\ge 2$ and the uniform interaction
distribution. In particular, this implies
almost sure polarization for all non-separable configurations 
in that case.

\paragraph{Outline}
In the rest of the paper, we proceed as follows. 
In \Cref{sec:preliminaries}, we introduce some notation and discuss several preliminaries. Then, we proceed sequentially, proving
the results discussed above.
In \Cref{sec:strictly-convex}, we show \Cref{lem:strictly-convex-intro}
for strictly convex configurations.
In \Cref{sec:active}, we prove polarization for active functions, i.e., \Cref{thm:active}.
\Cref{sec:stable-preliminaries} contains some general properties of the
dynamics of stable functions, as well as the proof of~\Cref{thm:three-ops} for three opinions.
In \cref{sec:2d}, we then prove polarization in two dimensions (\cref{thm:2d-main}). 
Finally, moving to the case of stable functions and $d,n\ge 3$, \Cref{sec:counterexample} contains the proof for \Cref{thm:counterexample}
and \Cref{sec:3d-stable} for \Cref{thm:dn-stable}.

\section{Preliminaries}
\label{sec:preliminaries}

\subsection{Setting and dynamics}
\label{sub_sec:dynamics}
Let \(d, n \geq 2\) denote the number of dimensions and the number of agents, respectively. 
We let \([n]\) denote the set \(\{1, 2, \hdots, n\}\) and refer to agents
by indices from this set. 
As mentioned,
    a \emph{configuration} \(\mathcal{U}\) is a tuple of \(n\) opinions in \(d\) dimensions, i.e., \(\mathcal{U} = (\vu_1, \hdots, \vu_n)\) where for every \(i\), \(\, \vu_i \in \mathbb{S}^{d-1}\).
    We denote
    the correlation \(A(\vu, \vec v)\) between two vectors \(\vu\) and \(\vec v\) as \(A(\vu,\vec v):=\bck{\vu, \vec v}\).
    Similarly, the (primary) angle between two opinions $\vu$ and $\vec v$ is denoted as
    $\alpha(\vu, \vec v):=\arccos(\langle \vu,\vec v\rangle)$. 
    In the context of $n$ opinions
    $(\vu_1,\ldots,\vu_n)$, we will write
    $A_{ij}:=A(\vu_i,\vu_j)$ and $\alpha_{ij} \coloneq\alpha(\vu_i, \vu_j)$.
Note that $0\le \alpha(\vu,\vec v)\le \pi$ for all opinions $\vu,\vec v$. We will also need
a related notion that we call \emph{effective angle}:
\begin{definition}
\label{def:effective-angle}
    The \emph{effective angle} of two opinions $\vu,\vec v$ is given as 
    $\gamma(\vu,\vec v):=\min(\alpha(\vu,\vec v),\pi-\alpha(\vu,\vec v))$.
    As with the correlations and angles, we write for short $\gamma_{ij}:=\gamma(\vu_i,\vu_j)$.
\end{definition}
Note that \(\alpha(\vu, \vec v) = \alpha(\vec v, \vu)\) and \(\gamma(\vu, \vec v) = \gamma(\vec v, \vu)\).
Also, the effective angle does not change when an opinion is flipped, but the angle may change; that is
\(\gamma(-\vu, \vec v) = \gamma(\vu, \vec v)\) while \(\alpha(-\vu, \vec v) = \pi - \alpha(\vu, \vec v)\).
Since the angle function $\alpha(\cdot,\cdot)$
is just the distance along the sphere, it satisfies
the triangle inequality
$\alpha(\vu_i,\vu_k)\le\alpha(\vu_i,\vu_j)
+\alpha(\vu_j,\vu_k)$. One can show that also the effective angle satisfies the triangle inequality, since it is the induced metric on the projective space
$\mathbb{RP}^{d-1}$.

Let us formally restate the dynamics of a single instance of our model
and derive a formula for how much a correlation changes due to one interaction.

\begin{definition}
\label{def:model}
    We call an element $(i,j)\in[n]\times[n]$ an
    \emph{interaction} and say that $j$ influences $i$.
    Let \(f:[-1, 1]\to \mathbb{R}\) be an update function
    and $\mathcal{D}$ a probability distribution on $[n]\times[n]$.
    
    We consider discrete timesteps \(t \in \{1, 2, \hdots\}\),
    and an infinite sequence of i.i.d.~random variables \(I^{(1)}, I^{(2)}, \hdots\) drawn from \(\mathcal{D}\).
    For a given initial configuration of opinions \(\mathcal{U}^{(0)}\), we define the configurations \(\mathcal{U}^{(t)}\) for \(t = 1, 2, \hdots\) as follows.
    
    If \(I^{(t)} = (i, j)\), then for every agent \(k\)
    \begin{align*}
        \vu_k^{(t+1)} &= \vu_k^{(t)} \, \text{ if } k \neq i \; , \\
        \vu_i^{(t+1)} &= \frac{\vw}{\norm{\vw}} \; \text{ where } \vw = \vu_i^{(t)} + f\left(A_{ij}^{(t)}\right) \cdot \vu_j^{(t)} \; ,
    \end{align*}
    where \(\vu_k^{(t)}\) is the opinion of agent \(k\) at time \(t\) and \(A_{ij}^{(t)}\) is the correlation between \(\vu_i^{(t)}\) and \(\vu_j^{(t)}\).
\end{definition}

\begin{remark}
    From the definition of the model, we can notice the following facts. If \(I^{(t)} = (i, j)\), then
    \begin{itemize}
        \item \(\vu_i^{(t+1)}\) is a linear combination of \(\vu_i^{(t)}\) and \(\vu_j^{(t)}\),
        \item for every \(k, l \in [n]\) such that \(k \neq i\) and \(l \neq i\), we have \, \(A_{kl}^{(t+1)} = A_{kl}^{(t)}\).
    \end{itemize}
\end{remark}

While it is in general possible that $\vw=0$ and therefore $\vu_i^{(t+1)}$ is 
ill-defined, \Cref{claim:new_correlation} below implies
that for active and stable update rules $\vw=0$ is not possible.

\begin{claim}[New correlation]
    \label{claim:new_correlation}
    Consider an update function $f$ and two opinions \(\vu_i, \vu_j\) with correlation \(A\).
    Let \(A'\) denote the new correlation after $j$ influences
    $i$.
    Then, we have $\|\vw\|^2=1+2Af(A)+f(A)^2$ and consequently
    \(A'\) is given by
    \begin{equation}
        \label{eq:new-correlation}
        A' = \frac{A + f(A)}{\sqrt{1 + 2 A f(A) + f(A)^2}} \; .
    \end{equation}
\end{claim}
\begin{proof}
    Indeed, we have $\vw=\vu_i+f(A)\vu_j$ and, by bilinearity,
    $\|\vw\|^2=\langle \vw,\vw\rangle=1+2Af(A)+f(A)^2$.
    Finally, indeed
    \begin{equation*}
        A'
        = \left\langle \frac{\vw}{\|\vw\|},\vu_j\right\rangle
        = \frac{A + f(A)}{\sqrt{1 + 2 A f(A) + f(A)^2}} \; .
        \qedhere
    \end{equation*}
\end{proof}

Since $\|\vw\|^2=1+2Af(A)+f(A)^2=(A+f(A))^2+(1-A^2)$, the only way for
$\vw=0$ to occur is $A\in\{-1,1\}$ and $A=-f(A)$.
However, our definitions of stable and active functions ensure
$f(1)>0$ and $f(-1)<0$ and therefore $f(1)\ne -1$ and $f(-1)\ne 1$.

\subsection{Transitivity properties}\label{subsec:transitivity_properties}
We now move to basic transitivity properties that follow
from triangle inequality for the Euclidean norm. These properties
are used throughout later proofs.

\begin{lemma}[Transitivity of closeness]
\label{lem:transitivity_of_closeness}
    Let \(d \geq 2\), \(0\le \epsilon < 1/4\), and consider three opinions \(\vu_i, \vu_j,\allowbreak \vu_k\).
    If \(\abs{A_{ij}} \geq 1 - \epsilon\) and \(\abs{A_{ik}} \geq 1 - \epsilon\), then
    \begin{equation*}
        \abs{A_{jk}} \geq 1 - 4 \epsilon \quad \text{ and } \quad \sign{(A_{jk})} = \sign{(A_{ij})} \sign{(A_{ik})} \; .
    \end{equation*}
\end{lemma}
\begin{proof}
    First, assume $\sign(A_{ij})=\sign(A_{ik})=1$. In that case, we
    need to show $A_{jk}\ge 1-4\eps$. Indeed, by assumption
    $\|\vu_i-\vu_j\|^2=\langle \vu_i-\vu_j,\vu_i-\vu_j\rangle=2(1-A_{ij})\le 2\eps$.
    Similarly, $\|\vu_i-\vu_k\|^2\le 2\eps$. By the triangle inequality,
    $\|\vu_j-\vu_k\|\le\|\vu_j-\vu_i\|+\|\vu_i-\vu_k\|\le\sqrt{8\eps}$.
    Hence, $A_{jk}=1-\|\vu_j-\vu_k\|^2/2\ge1-4\eps$.

    The other cases are reduced to the first case. For example,
    let $\sign(A_{ij})=-1$ and $\sign(A_{ik})=1$. We need to
    show $A_{jk}\le-(1-4\eps)$. Then, $A(\vu_i,-\vu_j)=-A_{ij}\ge 1-\eps$.
    Applying the first case to $\vu_i,-\vu_j,\vu_k$, we have
    $A(-\vu_j,\vu_k)\ge 1-4\eps$, hence $A_{jk}\le-(1-4\eps)$.
    The remaining cases proceed similarly.
\end{proof}
As a consequence, we obtain a simple criterion for strict convexity.
\begin{claim}
    \label{claim:convexity-of-one-cluster}
    Let \(0 \leq \epsilon < 1/4\). If \(\cU\) is such that for every opinions \(\vu_i,\vu_j\) we have \(\abs{A_{ij}} \geq 1 - \epsilon\), then \(\cU\) is strictly convex.
\end{claim}
\begin{proof}
    Let $b_1=1$ and $b_i=\sign(A_{1i})$ for $i>1$.
    That ensures $\langle b_1\vu_1,b_i\vu_i\rangle\ge 1-\eps$
    for every $i$.
    For general $i,j$, we have
    $\langle b_i\vu_i,b_j\vu_j\rangle\ge 1-4\eps>0$
    by~\cref{lem:transitivity_of_closeness}.
\end{proof}

\begin{lemma}[Inactivity relations]
\label{lem:transitivity_of_inactivity}
    Let \(d \geq 2\), $\eps\ge 0$, and consider three opinions \(\vu_i, \vu_j, \vu_k\).
    If \(\abs{A_{ij}} \leq \sqrt{\epsilon}\) and \(\abs{A_{ik}} \geq 1 - \epsilon\), then \(\abs{A_{jk}} \le 8 \sqrt{\epsilon}\).

    Furthermore, in the special case \(d = 2\),
    if \(\abs{A_{ij}} \leq \epsilon\) and \(\abs{A_{ik}} \leq \epsilon\), then \(\abs{A_{jk}} \geq 1 - 2 \epsilon^2\).
\end{lemma}

\begin{lemma}[Transitivity of activity]
\label{lem:transitivity_of_activity}
    Let \(d \geq 2\), $\eps_1, \epsilon_2 \ge 0$, and consider three opinions \(\vu_i, \vu_j, \vu_k\).
    If \(\abs{A_{ij}} \geq \sqrt{\epsilon_1}\) and \(\abs{A_{ik}} \geq 1 - \epsilon_2\), then \(\abs{A_{jk}} \geq \sqrt{\epsilon_1} - 3 \sqrt{\epsilon_2}\).
    In particular, if \(\epsilon_1 \geq 16 \epsilon_2\), then \(\abs{A_{jk}} \geq \sqrt{\epsilon_2}\).
\end{lemma}
\Cref{lem:transitivity_of_inactivity,lem:transitivity_of_activity} are proved 
%in the appendix. 
in Appendix \ref{appendix_A}. 
Let us also prove
\Cref{lem:cluster} about the clusters of $\eps$-inactive configurations.

\begin{proof}[Proof of \Cref{lem:cluster}]
Consider the relation between agents in $\mathcal{U}$ which is defined to hold 
if and only if the condition
$|A_{ij}|> 1/2$ holds. We first observe that it is an equivalence relation.
The only nontrivial part to check is transitivity.
Consider three agents \(i,j,k\) such that
$|A_{ij}|,|A_{ik}|> 1/2$, therefore by
$\eps$-inactivity $|A_{ij}|,|A_{ik}|\ge 1-\eps$.
Applying Lemma \ref{lem:transitivity_of_closeness}, 
indeed we have \(\abs{A_{jk}} \geq 1 - 4 \epsilon > 1/2\), where we used that $\epsilon<1/8$.

Since $\cU$ is $\eps$-inactive, $|A_{ij}|=1/2$ is impossible.
Therefore, each equivalence class of this relation is a cluster
and each cluster must be an equivalence class. Hence, we obtain
the desired partition.

It only remains to prove that there are at most $d$ clusters.
Assume otherwise and choose arbitrarily opinions $\vu_1,\ldots,\vu_{d+1}$, each from a different cluster.
By the Welch bounds~\cite{Welch1974}, there exist $i\neq j$ with $|A_{ij}|\geq 1/d$.
Therefore, the configuration cannot be $\eps$-inactive for any $\eps<1/d$.
\end{proof}

\subsection{Polarization into balanced groups}\label{subsec:polarization_into_balanced_groups}

In this subsection, we establish that, although consensus is a polarized configuration according to \Cref{def:polarization}, it occurs with vanishing probability if the update function is odd and the initial configuration is chosen according to a distribution fulfilling a symmetry condition.
This argument was already pointed out in~\cite[Remark 3.1]{HJMR}
in the setting of their model with external influences.
Here, we adapt it to our setting and state an explicit concentration bound.

\begin{definition}[Symmetric Distribution]
    A distribution \(\Gamma\) on \(\mathbb{S}^{d-1}\) is symmetric if for every subset \(A \subset \mathbb{S}^{d-1}\), we have \(\Gamma(A) = \Gamma(-A)\).
\end{definition}
\begin{lemma}
\label{lem:concentration}
    Let \(f\) be an odd update function, \(\cD\) the interaction distribution, \(\Gamma\) a symmetric distribution on \(\mathbb{S}^{d-1}\).
    Draw \(n\) opinions \(\vu_1, \hdots, \vu_n\) i.i.d. from the distribution \(\Gamma\) and let \(\cU^{(0)} = (\vu_1, \hdots, \vu_n)\).
    Let \(\cU^{(\infty)} = \lim_{t \to \infty} \cU^{(t)}\) when the limit exists.
    Then, for every \(c > 0\),
    \begin{equation*}
        \mathbf{P}_{\Gamma, \cD} \left[\abs{\;\left|\left\{i : \vu_i^{(\infty)} = \vu_1^{(\infty)}\right\}\right| - \frac{n+1}{2}} \geq c \sqrt{n-1} \; \biggr| \; \cU^{(t)} \text{ polarizes} \right] \leq 2 \exp(-2 c^2)\; .
    \end{equation*}
\end{lemma}

The bound in \Cref{lem:concentration} implies that under an odd update function, when opinions that are chosen from a symmetric distribution on \(\mathbb{S}^{d-1}\) polarize, then with high probability they form two antipodal groups of similar sizes.
Note that the expectation $(n+1)/2$
is slightly larger than $n/2$ since the first opinion is always equal to itself.

\begin{proof}
    Consider \(n\) independent random variables \(U_i \in \mathbb{S}^{d-1}\) with distribution \(\Gamma\),
    and \(n\) independent Rademacher random variables \(B_i \in \{-1, +1\}\) with uniform distribution.
    First, note that the random variables \(B_i U_i\) are also independent and distributed according to \(\Gamma\), because \(\Gamma\) is symmetric.

    Also, given that the function \(f\) is odd,
    for any particular initial configuration \(\cU^{(0)} = (\vu_1, \hdots, \allowbreak\vu_i, \hdots, \vu_n)\) and any sequence of interactions \(I^{(1)}, I^{(2)},\hdots\) such that \(\lim_{t\to\infty}\cU^{(t)}\) exists and is polarized,
    it holds that for \(\cU^{(0)} = (\vu_1, \hdots, -\vu_i, \hdots, \vu_n)\) and the same sequence of interactions the limit also exists and is a polarized configuration.
    The only difference is that if in the first case \(\vu_i^{(\infty)} = \vec v\), then in the second case \(\vu_i^{(\infty)} = -\vec v\).
    So, with \(\cV^{(0)} = (B_1 \vu_1, \hdots, B_n \vu_n)\), if the random
    processes are coupled using the same interactions
    $I^{(1)},I^{(2)},\ldots$, then $\cU^{(t)}$ polarizes if and only if
    $\cV^{(t)}$ polarizes. In particular, whether $\cV^{(t)}$ polarizes 
    does not depend on the Rademacher variables $B_i$.

    Accordingly, condition on some choice of $\cU^{(0)}$ and interactions
    $I^{(t)}$ such that $\mathcal{U}^{(t)}$ and $\mathcal{V}^{(t)}$ polarize,
    so that the only remaining randomness is over $(B_i)$.
    Let \(S = \left\{i : \vu_i^{(\infty)} = \vu_1^{(\infty)}\right\}\) 
    and \(S' = \left\{i : B_i \vu_i^{(\infty)} = B_1 \vu_1^{(\infty)}\right\}\). 
    We will now show
    \begin{equation*}
    \Pr\left[|S'|-\frac{n+1}{2}\ge c\sqrt{n-1}\right]\le 2 \exp(-2 c^2)\;.
    \end{equation*}
    The statement of the lemma then immediately follows by averaging over
    $\cU^{(0)}$ and $I^{(t)}$. As we have
    \begin{equation*}
        \abs{S'} = \sum_{i \in S} \mathbb{I}_{\left[B_i = B_1\right]} + \sum_{i \not\in S} \mathbb{I}_{\left[B_i = -B_1\right]} \; .
    \end{equation*}
    and as always \(1 \in S\), \(1 \in S'\),
    the random variable \(\abs{S'}-1\) is a binomial with parameters \((n-1, 1/2)\), therefore we obtain the desired inequality from Hoeffding's inequality, 
    \begin{equation*}
        \mathbf{P}_{\Gamma, \cD} \left[\abs{\abs{S'} -1 -\frac{n-1}{2}} \geq c \sqrt{n-1} \right] \leq 2 \exp(-2c^2) \; . \qedhere
    \end{equation*}
\end{proof}

\section{Strictly convex polarization: Proof of \Cref{lem:strictly-convex-intro}}
\label{sec:strictly-convex}
In this section, we prove \Cref{lem:strictly-convex-intro}, that
is, we show that strictly convex
configurations polarize almost surely.
In order for our proof to be applicable
for both active and stable update functions,
we introduce a definition:

\begin{definition}
    Let $c\ge 0, 0\le\beta<1$.
    An update function $f$ is
    \emph{$(c,\beta)$}-contractive if,
    for every $A$ with $|A|>c$:
    \begin{enumerate}
        \item $\sign(f(A))=\sign(A)$.
        \item $1-|A'|\le\beta\cdot(1-|A|)$,
        where $A' = (A + f(A))/\sqrt{1 + 2 A f(A) + f(A)^2}$
    is the new correlation (see \Cref{claim:new_correlation}).
    \end{enumerate}
\end{definition}

We will now establish the following sequence of results:

\begin{lemma}[Strictly convex polarization]
\label{lem:convergence-of-c-strictly-convex-configuration}
    Let $c\ge 0,0\le\beta<1$ and $\mathcal{U}^{(0)}$ be $c$-strictly convex. If the update function is $(c,\beta)$-contractive,
    then $\mathcal{U}^{(t)}$ polarizes almost surely.
\end{lemma}

\begin{claim}[Contraction for active functions]
\label{cl:contraction-active}
Let $f$ be an active update function.
Recall that there exists some $A_0$ with $|A_0|<1$ such
that $f(A)>0$ for $A>A_0$ and 
$f(A)<0$ for $A<A_0$. 
There exists $0<\beta<1$ such that the following holds. 
Let $A$ be a correlation between two vectors
and let $A'$ be the updated correlation after one
interaction between them. If $A>A_0$, then
\(1-A'\le\beta(1-A)\).
Similarly, if $A<A_0$, then $A'-(-1)\le\beta(A-(-1))$.

In particular, $f$ is $(|A_0|,\beta)$-contractive.
\end{claim}

\begin{claim}[Contraction for stable functions]
\label{claim:contraction_stable}
    Let $f$ be a stable update function.
    For every $\delta>0$, there exists $\beta=\beta(\delta) \in (0,1)$ such that 
    $f$ is $(\delta,\beta(\delta))$-contractive.
\end{claim}

It is easy to see that \Cref{lem:strictly-convex-intro}
follows from
\Cref{lem:convergence-of-c-strictly-convex-configuration}, \Cref{cl:contraction-active},
and \Cref{claim:contraction_stable}.
(In the stable case, note that if $\cU^{(0)}$
is strictly convex, then it is $\delta$-strictly
convex for some $\delta>0$.) We prove
\Cref{lem:convergence-of-c-strictly-convex-configuration} in \Cref{sec:proof-for-c-strictly-convex}, 
\Cref{cl:contraction-active} in
%the appendix,
Appendix~\ref{appendix_A}, 
and
\Cref{claim:contraction_stable}
in \Cref{sec:contraction-stable-proof}.

\subsection{Proof of \Cref{lem:convergence-of-c-strictly-convex-configuration}}
\label{sec:proof-for-c-strictly-convex}

Recall the outline from \Cref{sec:proof-outline}.
As mentioned there,
we will utilize a potential function
    $\mathcal{P}(t)=1-\min_{i,j\in[n]}\left|A_{ij}^{(t)}\right|$.
    \Cref{lem:convergence-of-c-strictly-convex-configuration} follows from the
    following intermediate results.
    
    \begin{lemma}
        \label{claim:minimum-correlation-non-decreasing}
        Assume \(\mathcal{U}^{(t)}\) is \(c\)-strictly convex, $f$ is $(c,\beta)$-contractive and that, at time step $t$, agent $i$ is influenced by agent $j$. Then, for every $k\neq i$, 
        \begin{equation}
            \label{eq:agent-potential-non-decreasing}
            \min_{1\leq l \leq n} |A_{kl}^{(t+1)}| \geq \min_{1\leq l \leq n} |A_{kl}^{(t)}|.
        \end{equation}
        Furthermore, $\mathcal{P}(t+1)\le
        \mathcal{P}(t)$ and the configuration
        $\mathcal{U}^{(t+1)}$ is also
        $c$-strictly convex with the same
        values of $b_1,\ldots,b_n$.
    \end{lemma}

    \begin{lemma}
        \label{cl:n-choose-two}
        Let $d,n\ge 2$, $\mathcal{D}$ be a fully supported distribution and $f$
        a $(c,\beta)$-contractive
        update function. There exists $p>0$,
        such that, for any
        $c$-strictly convex configuration $\cU^{(0)}$
        and $t=\binom{n}{2}$,
        \begin{equation}
            \Pr \left[ \mathcal{P}(t) \leq \beta \cdot
            \mathcal{P}(0)\right] \geq p.
        \end{equation}
    \end{lemma}

    \begin{lemma}
        \label{lem:potential-implies-polarization}
        Let $(\mathcal{U}^{(t)})_t$ be a sequence of configurations
        such that:
        \begin{enumerate}
            \item 
        $\mathcal{U}^{(0)}$ is $c$-strictly convex.
        \item The update function
        $f$ is $(c,\beta)$-contractive for
        some $\beta<1$.
        \item
        $\lim_{t\to\infty}\mathcal{P}(t)=0$. 
        \end{enumerate}
        Then,
        $\mathcal{U}^{(t)}$ polarizes.
    \end{lemma}

Let us show \Cref{lem:convergence-of-c-strictly-convex-configuration}
assuming that \Cref{lem:potential-implies-polarization,claim:minimum-correlation-non-decreasing,cl:n-choose-two} hold. Then, we will
prove those lemmas in turn.

\begin{proof}[Proof of \Cref{lem:convergence-of-c-strictly-convex-configuration}]
    Consider a sequence of Bernoulli random
    variables $X_0,\ldots,X_k,\ldots$, where $X_k$ is the indicator
    of event $\mathcal{P}\left({(k+1)\binom{n}{2}}\right)\le \beta\cdot \mathcal{P}\left(k\binom{n}{2}\right)$.
    By \Cref{claim:minimum-correlation-non-decreasing},
    the configuration $\mathcal{U}^{(t)}$ remains
    $c$-strictly convex at all times.
    \Cref{cl:n-choose-two} implies that $\Pr[X_k=1]\ge p$ for some constant $p>0$
    and for every $k$, even conditioned on $(X_i)_{i<k}$.
    Applying, e.g., the second Borel--Cantelli lemma, almost surely
    $X_k=1$ happens for infinitely many values of $k$. 
    Since, by \Cref{claim:minimum-correlation-non-decreasing}, $\mathcal{P}(t+1)\le\mathcal{P}(t)$
    holds in every case, if $X_k=1$ for infinitely
    many $k$, then $\lim_{t\to\infty}
    \mathcal{P}(t)=0$.
    But that in turn
    implies the polarization of $\mathcal{U}^{(t)}$ by
    \Cref{lem:potential-implies-polarization}.
\end{proof}

\begin{proof}[Proof of \Cref{claim:minimum-correlation-non-decreasing}]
    Let $b_1,\ldots,b_n$ be the values from
    the definition of strict convexity for
    configuration $\mathcal{U}^{(t)}$.
    Assume that, at time step $t$, agent $i$ is influenced by agent $j$ and let $k\neq i$.
    Observe that due to the definition 
    of $(c,\beta)$-contractivity
    for every $i,j$ 
    we have the property $\sign(A_{ij}^{(t)})=\sign(f(A_{ij}^{(t)}))$.
    
    When \(l \neq i\), we have $A_{kl}^{(t+1)} = A_{kl}^{(t)}$. 
    Otherwise, observe that
    \begin{equation}
    \label{eq:sign-convex-configuration}
        \sign \left( f(A_{ij}^{(t)}) A_{jk}^{(t)} \right) = 
        \sign\left(A_{ij}^{(t)}\right)
        \sign\left(A_{jk}^{(t)}\right)
        =(b_i b_j) \cdot (b_j b_k) = b_i b_k = \sign \left( A_{ik}^{(t)} \right).
    \end{equation}
    Since 
    \begin{equation*}
        A_{ik}^{(t+1)}
        = \bck{\frac{\vu_i^{(t)} + f(A_{ij}^{(t)}) \vu_j^{(t)} }{\norm{\vu_i^{(t)} + f(A_{ij}^{(t)}) \vu_j^{(t)} }},   \vu_k^{(t)}} 
        = \frac{ A_{ik}^{(t)}+f(A_{ij}^{(t)}) A_{jk}^{(t)}}{\sqrt{1+2A_{ij}^{(t)}f(A_{ij}^{(t)})+f^2(A_{ij}^{(t)})}}\;,
    \end{equation*}
    it follows from~\eqref{eq:sign-convex-configuration}
    that $\sign(A_{ik}^{(t+1)})=\sign(A_{ik}^{(t)})$
    and furthermore
    \begin{align*}
        \left|A_{ik}^{(t+1)}\right|
        &\ge
        \frac{\left| A_{ik}^{(t)} \right|+\left| f(A_{ij}^{(t)})\right| \cdot \left| A_{jk}^{(t)}\right| }{\sqrt{1+2A_{ij}^{(t)}f(A_{ij}^{(t)})+f^2(A_{ij}^{(t)})}}\\
        &\geq \frac{1}{1+\left|f(A_{ij}^{(t)})\right|} \left( \left| A_{ik}^{(t)} \right|+\left| f(A_{ij}^{(t)})\right| \cdot \left| A_{jk}^{(t)}\right|  \right)\\
        &\geq \frac{1}{1+\left|f(A_{ij}^{(t)})\right|} \left(\min_{1\leq l \leq n} \left| A_{kl}^{(t)} \right| + \left|f(A_{ij}^{(t)})\right| \min_{1\leq l \leq n} \left| A_{kl}^{(t)} \right| \right)
        = \min_{1\leq l \leq n} \left| A_{kl}^{(t)} \right|,
    \end{align*}
    which completes the proof of \eqref{eq:agent-potential-non-decreasing}.
    In particular,
    \begin{equation}
        \label{eq:minimum-correlation-non-decreasing}
        \min_{k,l\in [n]} \left| A_{kl}^{(t+1)} \right| = \min_{k\neq i} \min_{l \in [n]}\left| A_{kl}^{(t+1)} \right|
        \overset{\text{\eqref{eq:agent-potential-non-decreasing}}}{\geq}
        \min_{k\neq i} \min_{l \in [n]}\left| A_{kl}^{(t)} \right| = \min_{k,l\in [n]} \left| A_{kl}^{(t)} \right|
    \end{equation}
    and $\mathcal{P}(t+1)\le\mathcal{P}(t)$
    follows.

    Finally, the calculations above imply 
    $\sign(A_{ij}^{(t+1)})=\sign(A_{ij}^{(t)})$
    and 
    $|A_{ij}^{(t+1)}|\ge 1-\mathcal{P}(t+1)
    \ge 1-\mathcal{P}(t)>c$ for every $i,j$,
    therefore also
    \(\bck{b_i \vu_i^{(t+1)}, b_j \vu_j^{(t+1)}} > c\),
    and the configuration remains $c$-strictly convex
    with the same $b_1,\ldots,b_n$.
\end{proof}

\begin{proof}[Proof of \Cref{cl:n-choose-two}]
    Since the interaction distribution $\mathcal{D}$ has full support, it suffices to give a sequence of interactions of length $\binom{n}{2}$ such that
    $\mathcal{P}\left(\binom{n}{2}\right)\le
    \beta\cdot\mathcal{P}(0)$.
    Since by \Cref{claim:minimum-correlation-non-decreasing} the configuration remains $c$-strictly convex, and by $(c,\beta)$-contractivity,
    if agent $i$ is influenced by agent $j$ at time $t$, then 
    \begin{equation}\label{eq:using-claim-2}
    1-|A_{ij}^{(t+1)}|\leq \beta\left( 1- |A_{ij}^{(t)}|\right).
    \end{equation}
    We claim that the following sequence of interactions is as desired: For all $i \in \{1,\dots,n-1\}$ in increasing order, let agent $i$ influence all agents $j$ with $j > i$.
    Say that, in the procedure above, agent $i$ 
    influences other agents from time $t_1$ to $t_2:=t_1+n-i$.
    Since \eqref{eq:using-claim-2} applies to every agent $j>i$, we have 
    \begin{equation*}
    1-\min_{j:j> i} |A_{ij}^{(t_2)}| \leq \beta\left( 1- \min_{j:j>i}|A_{ij}^{(t_1)}|\right)
    \le \beta\left(1-\min_{j,k}|A_{jk}^{(0)}|\right)\;.
    \end{equation*}
    By \Cref{claim:minimum-correlation-non-decreasing}
    applied to agents $\{i,\ldots,n\}$, the value 
    $\min_{j:j>i} |A_{ij}^{(t)}|=\min_{j:j\ge i}|A_{ij}^{(t)}|$ does not decrease
    for $t\ge t_2$ because agent~$i$ is not influenced again. Therefore,
    \begin{equation*}
    1-\min_{i,j} \left| A_{ij}^{\left(\binom{n}{2}\right)}\right| = 1-\min_{1 \leq i \leq n} \min_{j:j> i} \left|A_{ij}^{\left(\binom{n}{2}\right)} \right| 
    \leq \beta\left( 1- \min_{i,j} \left|A_{ij}^{(0)} \right|\right),
    \end{equation*}
    which as explained before proves the claim.
\end{proof}

\begin{proof}[Proof of \Cref{lem:potential-implies-polarization}]
    Let $b_1,\ldots,b_n$ be signs from the definition of strict convexity
    and let $1\le i\le n$. We will first show that
    sequence $\left(\vu_{i}^{(t)}\right)_t$ is Cauchy, and therefore converges.
    By way of contradiction, assume that the sequence is not Cauchy.
    That is, there exists $\eps>0$ such that for all $T$
    there exist $t_1,t_2>T$ with
    $\left\|\vu_i^{(t_1)}-\vu_i^{(t_2)}\right\|>\eps$.
    Since the norm is a continuous function of the angle, 
    there exists $\eps'>0$ such that it also holds
    $\alpha\left(\vu_i^{(t_1)},\vu_i^{(t_2)}\right)>\eps'$.
    
    However, by assumption on convergence of $\mathcal{P}(t)$,
    for $t_1$ large enough we have
$\alpha\big(b_j\vu_j^{(t_1)},\allowbreak b_{j'}\vu_{j'}^{(t_1)}\big)<\eps'/2$ for every $1\le j,j'\le n$.
    Due to $\sign(f(A))=\sign(A)$ for $|A|>c$,
    at any future time $t_2>t_1$, vector $b_i\vu_i^{(t_2)}$ lies in the projection onto the sphere of the convex
    hull of $(b_1\vu_1^{(t_1)},\ldots,\allowbreak b_n\vu_n^{(t_1)})$.
    It follows $\alpha\left(\vu_i^{(t_2)},\vu_i^{(t_1)}\right)<\eps'/2$, which is a contradiction.
    Accordingly, sequence $\left(\vu_{i}^{(t)}\right)_t$ is Cauchy
    and converges.
    
    Furthermore, since $\mathcal{P}(t)$ converges to zero, and using \Cref{claim:minimum-correlation-non-decreasing},
    for every $i<j$
    it holds \\
    $\lim_{t\to\infty}A\left(b_i\vu_i^{(t)},b_j\vu_j^{(t)}\right)=1$. Therefore the limits of sequences $b_i\vu_i^{(t)}$
    are all equal and $\mathcal{U}^{(t)}$
    polarizes.
\end{proof}

\subsection{Proof of \Cref{claim:contraction_stable}}
\label{sec:contraction-stable-proof}
\begin{subequations}
Consider two opinions 
with correlation \(A\) and effective angle \(\gamma\).
In this section let us use the notation
\(A'\) and \(\gamma'\)
for the new correlation and effective angle
after an interaction between them.
First, recall from \Cref{claim:new_correlation} that the new correlation is given by
    \begin{equation}
        \label{eq:73}
        A' = \frac{A + f(A)}{\sqrt{1 + 2 A f(A) + f(A)^2}} \; .
    \end{equation}    

\begin{claim}
\label{cor:A_is_non_decreasing}
    If \(f\) is stable, then $\sign(A')=\sign(A)$. Furthermore,
    \begin{align}
        \label{eq:74}
        A \in \{-1, 0, 1\} &\implies A' = A \; ,
        \\
        \label{eq:75}
        A \not\in \{-1, 0, 1\} &\implies 1 > \abs{A'} > \frac{|A+f(A)|}{1+|f(A)|}> \abs{A} > 0 \; .
    \end{align}
\end{claim} 
\end{subequations}
\begin{proof}
    \(f\) is stable, so by definition \(\sign(A) = \sign(f(A))\), which implies with \eqref{eq:73} that \(\sign(A') = \sign(A)\). Let's then show \eqref{eq:74} and \eqref{eq:75}.
    The case $A\in\{-1,0,1\}$ is easy to check directly. Otherwise,
    we know that \(-1 < A < 1\), and \(A f(A) < \abs{f(A)}\). Therefore \(1 + 2 A f(A) + f(A)^2 < \left(1 + \abs{f(A)}\right)^2\). Hence,
    \begin{equation*}
        \abs{A'}
        > \frac{\abs{A + f(A)}}{1 + \abs{f(A)}}
        = \frac{\abs{A} + \abs{f(A)}}{1 + \abs{f(A)}}
        >\frac{|A|(1+|f(A)|)}{1+|f(A)|}=|A|\; .
    \end{equation*} 
    Furthermore, from the definition of the model an interaction between two non co-linear opinions cannot make them co-linear, so \(\abs{A'} < 1\).
\end{proof}

\begin{proof}[Proof of~\cref{claim:contraction_stable}]
    We define $\beta:=\max \left\{ 1/(1+|f(x)|): x \in [-1,-\delta] \cup [\delta,1]\right\}$. Observe that the maximum exists because $f$ is continuous and $[-1,-\delta] \cup [\delta,1]$ is compact. Moreover, observe that $\beta<1$ because, since $f$ is stable, we have $f(x)\neq 0$ for $x \neq 0$.
    We obtain
    \begin{equation*}
        |A'| \overset{\text{Cl.~\ref{cor:A_is_non_decreasing}}}{\geq} \frac{|A + f(A)|}{1 + |f(A)|} 
        \overset{\text{sign}(A)=\text{sign}f(A)}{=} \frac{|A| + |f(A)|}{1 + |f(A)|} 
        = 1- \frac{1-|A|}{1+|f(A)|} \geq 1-\beta(1-|A|).
    \end{equation*}
    This implies that \(f\) is
    $(\delta,\beta)$-contractive.
\end{proof}

For future use, we give a similar contraction result for the effective angle.
The proof is included in
%the appendix.
Appendix~\ref{appendix_A}.

\begin{claim}[Contraction of the effective angle]
\label{claim:3}
    Let $f$ be a stable update function.
    For every $\delta>0$, there exists \(0 < c_0 < 1\) such that, if \(\gamma \leq \frac{\pi}{2}-\delta\), then \(\gamma' \leq c_0 \gamma\).
\end{claim}

\section{Active update functions: Proof of \Cref{thm:active}}
\label{sec:active}

In this section we prove
\Cref{lem:almost_surely_c_strictly_convex} and \Cref{thm:active}.

\begin{proof}[Proof of \Cref{lem:almost_surely_c_strictly_convex}]
        The strategy of the proof is to make agent \(1\) influence every other agent until all opinions are so close to each other that the configuration is \(c\)-strictly convex,
        for $c=|A_0|$ where $A_0$
        is taken from \Cref{def:active_function}.
        To that end, choose some \(\epsilon < \min(1 -c,1/256)\), and let \(\epsilon_0 = \epsilon / 4\).
        If for an opinion \(\vu_i\) it holds \(A_{1i} = A_0\), then after one interaction where 1 influences $i$ it holds \(A_{1i} \neq A_0\) and the proof reduces to one of the other cases. If \(A_{1i} > A_0\), then from \Cref{cl:contraction-active}, in at most \(T = \lceil\log_{\beta}(\epsilon_0/2)\rceil\) steps of agent \(1\) influencing \(i\) we will have \(A_{1i} \geq 1 - \epsilon_0\). Similarly, if \(A_{1i} < A_0\), then 
        in at most \(T\) steps of agent \(1\) influencing \(i\) we will have \(A_{1i} \leq - 1 + \epsilon_0\).
        
        This implies that if we make agent \(1\) influence every other agent \(T+1\) times each, then after \(K = (n-1)(T+1)\) interactions we will have
        \(\abs{A_{1i}^{(K)}} \geq 1 - \epsilon_0\) for every \(i \in [n]\). From \Cref{lem:transitivity_of_closeness}, for every \(i, j \in [n]\), \(\abs{A_{ij}^{(K)}} \geq 1 - 4\epsilon_0 \geq 1 - \epsilon\).
        Therefore, by \Cref{claim:convexity-of-one-cluster}, \(\cU^{(K)}\) is strictly convex. Furthermore, for every \(i, j\), \(\abs{A_{ij}^{(K)}} \geq 1 - \epsilon > c\). Therefore, \(\cU^{(K)}\) is \(c\)-strictly convex.
\end{proof}

\begin{proof}[Proof of \Cref{thm:active}]
From \Cref{lem:almost_surely_c_strictly_convex},
    at any time $t$, independently of the past,
    if configuration $\mathcal{U}^{(t)}$ is not $|A_0|$-strictly convex,
    then
    \begin{equation*}
    \Pr\left[\mathcal{U}^{(t+K)}\text{ is $|A_0|$-strictly convex}\;\vert\; \mathcal{U}^{(t)}\right]\ge p_{\min}^{K}\;,
    \end{equation*}
    where \(p_{\min}>0\) is the minimum probability in \(\mathcal{D}\). 
    That implies that almost surely, 
    eventually $\mathcal{U}^{(t)}$ becomes $|A_0|$-strictly convex.
    Then, by \Cref{lem:strictly-convex-intro},
    $\mathcal{U}^{(t)}$ polarizes.
\end{proof}

\section{Stable update functions and proof of \Cref{thm:three-ops}}
\label{sec:stable-preliminaries}
From now on, we study polarization under stable update functions.
We begin the discussion in \Cref{sec:properties_of_stable_update_functions}
with observing that interactions under stable
functions preserve active, separable, and strictly
convex properties.
In \Cref{sec:three-opinions}, we show that three opinions
polarize almost surely (\Cref{thm:three-ops}). 
Finally, in \Cref{sec:more_than_three_opinions} we
show that every $\eps$-active configuration
eventually becomes $\eps$-inactive.
This result is used in
\Cref{sec:2d} and \Cref{sec:3d-stable}.

\subsection{Properties of stable update functions}
\label{sec:properties_of_stable_update_functions}

In this section we give a few properties of various types of configurations.
Recall active, separable and strictly convex configurations from \Cref{def:epsilon_activity,def:separable,def:strictly_convex}.
We show that those properties are
preserved over time. The first claim follows immediately from
\Cref{cor:A_is_non_decreasing}.

\begin{claim}
\label{claim:4}
    Let $f$ be a stable update function.
    $\mathcal{U}^{(t)}$ is active if and only if $\mathcal{U}^{(t+1)}$ is active.
\end{claim}

\begin{lemma}
\label{lem:non-separable-preserved}
    Let $f$ be a stable update function.
    A configuration \(\mathcal{U}^{(t+1)}\) is separable if and only if \(\mathcal{U}^{(t)}\) is separable.
\end{lemma}
\begin{proof}
    Let us prove the first implication.
    Assume that $\mathcal{U}^{(t)}$ is separable.
    From \Cref{def:separable}, there exists a subspace \(V \subset \mathbb{R}^d\) such that for every opinion \(\vu_i\), either \(\vu_i \in V\) or \(\vu_i \in V^T\), where \(V^T\) is the orthogonal complement of \(V\).

    Let the configuration at time $t+1$ be obtained by \(j\) influencing \(i\).
    Given that \(\vu_i^{(t+1)}\) is a linear combination of \(\vu_i^{(t)}\) and \(\vu_j^{(t)}\), we have that
    \begin{itemize}
        \item if \(\vu_i^{(t)}, \vu_j^{(t)} \in V\) (respectively \(\vu_i^{(t)}, \vu_j^{(t)} \in V^T\)), then \(\vu_i^{(t+1)} \in V\) (respectively \(\vu_i^{(t+1)} \in V^T\)),
        \item if \(\vu_i^{(t)} \in V\) and \(\vu_j^{(t)} \in V^T\) (or vice versa), then \(A_{ij}^{(t)} = 0\), so the interaction has no effect and \(\mathcal{U}^{(t+1)} = \mathcal{U}^{(t)}\).
    \end{itemize}
    We conclude that \(\mathcal{U}^{(t+1)}\) is also separable.

    Conversely,
    let's show that if \(\mathcal{U}^{(t+1)}\) is separable, then so is \(\mathcal{U}^{(t)}\). Suppose for the sake of contradiction that \(\mathcal{U}^{(t+1)}\) is separable and \(\mathcal{U}^{(t)}\) is not. So there exists \(V\) such that for every \(i \in [n]\), either \(\vu_i^{(t+1)} \in V\) or \(\vu_i^{(t+1)} \in V^T\). Given that \(\cU^{(t)}\) is not separable, if at time \(t\), \(j\) influenced \(i\), it means \(\vu_i^{(t)} \not\in V\) and \(\vu_i^{(t)} \not\in V^T\), while \(\vu_j^{(t)} \in V\) or \(\vu_j^{(t)} \in V^T\). Furthermore, the opinion of $i$ mush
    have moved, so $A_{ij}^{(t)}\ne 0$.

    Assume $\vu_j^{(t)}\in V$.
    An interaction between \(i\) and \(j\) cannot make their opinions orthogonal (\Cref{cor:A_is_non_decreasing}) so \(\vu_i^{(t+1)} \not\in V^T\). 
    But
    $\vu_{i}^{(t+1)}$ is a non-trivial linear combination
    of $\vu_i^{(t)}\notin V$ and $\vu_{j}^{(t)}\in V$, so it cannot
    lie in $V$ either, a contradiction.
    A similar argument works for the case \(\vu_j^{(t)} \in V^T\).
\end{proof}

\begin{claim}
\label{cl:stable-convex-preserved}
Let $f$ be stable and $\mathcal{U}^{(t)}$ strictly convex for
some $b_1,\ldots,b_n$. Then, $\mathcal{U}^{(t+1)}$ is strictly convex
for the same $b_1,\ldots,b_n$.
\end{claim}
\begin{proof}
Since $\mathcal{U}^{(t)}$ is strictly convex, it
is $\delta$-strictly convex for some $\delta>0$. The result follows
by \Cref{claim:contraction_stable} and \Cref{claim:minimum-correlation-non-decreasing}.
\end{proof}

\subsection{Polarization of three opinions: Proof of \Cref{thm:three-ops}}
\label{sec:three-opinions}

Let $f$ be a fixed stable update function.
Let $\mathcal{U}^{(0)}$ be 
a starting configuration of three opinions
for some $d\ge 2$ 
which
is not separable (since opinions always
stay in the span of initial opinions, in fact we could assume w.l.o.g.~that $d=3$, however, we do not use this assumption in
the proof).
Recall the effective angle from \Cref{def:effective-angle}.
Given configuration $\mathcal{U}^{(t)}$ of three opinions at time $t$, let
\begin{equation*}
P^{(t)}:=\gamma_{12}^{(t)}+\gamma_{13}^{(t)}+\gamma_{23}^{(t)}\;.
\end{equation*}

Suppose \(\cU^{(t)}\) is not separable. 
The first step in the proof will be showing that $P^{(t)}$ never increases
with $t$:
\begin{lemma}\label{lem:potential-monotone}
Almost surely, $P^{(t+1)}\le P^{(t)}$. 
\end{lemma}

\begin{proof}
Assume without loss of
generality that agent 2 influences agent 1. Since only opinion 1 moves, we have to
show $\gamma_{12}^{(t+1)}+\gamma_{13}^{(t+1)}\le \gamma_{12}^{(t)}+\gamma_{13}^{(t)}$.
First, let us show two inequalities
\begin{subequations}
\begin{align}
\gamma_{12}^{(t+1)}+\alpha_{13}^{(t+1)}&\le\gamma_{12}^{(t)}+\alpha_{13}^{(t)}\;,
\label{eq:09}\\
\gamma_{12}^{(t+1)}+\alpha_{13}^{(t)}&\le\gamma_{12}^{(t)}+\alpha_{13}^{(t+1)}\;.
\label{eq:10}
\end{align}
\end{subequations}
To establish them, recall
that the angular distance $\alpha(\cdot,\cdot)$ is a metric on $\mathbb{S}^{d-1}$,
in particular it obeys the triangle inequality. Now, observe that indeed
\begin{align*}
\alpha_{13}^{(t+1)}
&=\alpha(\vu_1^{(t+1)},\vu_3^{(t)})
\le\alpha(\vu_1^{(t+1)},\vu_1^{(t)})+\alpha(\vu_1^{(t)},\vu_3^{(t)})
=\gamma_{12}^{(t)}-\gamma_{12}^{(t+1)}+\alpha_{13}^{(t)}\;,\\
\alpha_{13}^{(t)}
&=\alpha(\vu_1^{(t)},\vu_3^{(t)})
\le\alpha(\vu_1^{(t)},\vu_1^{(t+1)})+\alpha(\vu_1^{(t+1)},\vu_3^{(t)})
=\gamma_{12}^{(t)}-\gamma_{12}^{(t+1)}+\alpha_{13}^{(t+1)}\;.
\end{align*}
Inequalities~\eqref{eq:09} and~\eqref{eq:10} follow by rearranging.
Finally, if $\alpha_{13}^{(t)}\le\pi/2$, then
$\gamma_{12}^{(t+1)}+\gamma_{13}^{(t+1)}\le
\gamma_{12}^{(t+1)}+\alpha_{13}^{(t+1)}\le
\gamma_{12}^{(t)}+\alpha_{13}^{(t)}=
\gamma_{12}^{(t)}+\gamma_{13}^{(t)}$.
On the other hand, if $\alpha_{13}^{(t)}>\pi/2$, then
$\gamma_{12}^{(t+1)}+\gamma_{13}^{(t+1)}\le\gamma_{12}^{(t+1)}+\pi-\alpha_{13}^{(t+1)}
\le\gamma_{12}^{(t)}+\pi-\alpha_{13}^{(t)}=\gamma_{12}^{(t)}+\gamma_{13}^{(t)} $.
\end{proof}

Since $P^{(t)}$ is bounded and nonincreasing with $t$, 
the limiting random variable $P^{\infty}$ is well defined.
One could think that we are almost done, but this is not quite the case.
The state space is uncountable and therefore
it is possible that the potential decreases
``too slowly'' to converge to zero.
In particular, \Cref{lem:potential-monotone}
applies to
a separable configuration
with $\vu_1\perp\{\vu_2,\vu_3\}$, but for that configuration
it will obtain $\lim_{t \to\infty} P^{(t)}=\pi$.
We also need to handle configurations that do polarize, 
but where every interaction makes only small progress. 
To achieve that, we will need an additional
step.

First, let us show that for large enough angles we do
make constant progress:

\begin{claim} \label{claim:52}
For every $\delta>0$, there exists $k$ such that the following holds.
If $1.2\pi\le P^{(t)}\le 3\pi/2-\delta$, then
there exists a sequence of $k$ interactions such that
$P^{(t+k)}\le P^{(t)}-0.05\pi$.
\end{claim} 
\begin{proof}
Let us rename the opinions if necessary so that
$\gamma_{12}^{(t)}\le\gamma_{13}^{(t)}\le\gamma_{23}^{(t)}$.
Since $1.2 \pi \leq P^{(t)} \leq 3\pi/2-\delta$, it follows $0.2\pi\le\gamma_{12}\le\pi/2-\delta/3$
and $\gamma_{23}\ge 0.4\pi$. In particular,
$\gamma_{12}$ is strictly bounded away from $\pi/2$ and zero. Take $c_0$ from \cref{claim:3} invoked for $\delta/3$. 
This means we can let agent 1 influence agent 2 a bounded number of steps $k_a=\log(1-0.15\pi/\gamma_{12})\big/\log(c_{0})$ until $\gamma_{12}$ decreases by $0.15\pi$. Since $\gamma_{23}$ is at least $0.4\pi$ it can only increase by $0.1\pi$ so that the potential decreases by at least $0.05\pi$.
\end{proof}
Second, we show that any other configuration
can become strictly convex after a constant number of interactions:
\begin{claim}
There exists $k$ such that, for every configuration that is not separable and with
$P^{(t)}\le 1.2\pi$, there is a sequence of $k$ interactions
that make $\mathcal{U}^{(t+k)}$ strictly convex.
\label{claim:toconv}
\end{claim}
\begin{proof}
In the following for clarity let us drop the superscript
for values at times $t$, so that $\gamma_{12}=\gamma_{12}^{(t)}$, 
etc. Without loss of generality, assume that 
$\gamma_{12}\le\gamma_{13}\le\gamma_{23}$.
Since $\mathcal{U}^{(t)}$ is not separable, it holds
$\gamma_{13}<\pi/2$, or in other words $|A_{13}|>0$,
since otherwise opinion 3 would be orthogonal
to both 1 and 2. Furthermore, observe that $P^{(t)}\le 1.2\pi$
implies $\gamma_{12}\le 0.4\pi$. By \Cref{claim:3},
influencing agents 1 and 2
a bounded number of times if necessary, let us also assume from now
on that 
$\gamma_{12}\le 0.01\pi$. 

We will now show that the configuration becomes
strictly convex if, $K$ times in a row, agent $1$ influences agent 
$2$ (for large enough $K$). To that end,
let $\tilde{\vec v}$ be the projection of $\vu_3$ onto the span
of $\vu_1$ and $\vu_2$ and let $\vec v=\tilde{\vec v}/\|\tilde{\vec v}\|$.
For any opinion $\vu$ in the span of $\vu_1$ and $\vu_2$, clearly
it holds $A(\vu,\vu_3)=A(\vu,\tilde{\vec v})$. In particular, we have
$\sign(A(\vu,\vu_3))=\sign(A(\vu,\vec v))$.

Let us consider an easier special case first.
Assume that $\gamma(\vu_1,\vec v)\le 0.05\pi$.
Let us show that in that case the configuration
is already strictly convex. 
Let
$b_2=\sign(A_{12})$ and $b_3=\sign(A_{13})$.
We need to check that the three angles
$\alpha(\vu_1,b_2\vu_2)$, $\alpha(\vu_1,b_3\vu_3)$ and
$\alpha(b_2\vu_2,b_3\vu_3)$ are strictly less than $\pi/2$.

First, $\alpha(\vu_1,b_2\vu_2)=\gamma_{12}\le 0.01\pi$. 
Furthermore, by what we said above
we have $A(\vu_1,\vec v)\ne 0$, which implies
$A_{13}\ne 0$ and $\alpha(\vu_1,b_3\vu_3)=\gamma_{13}<\pi/2$.
Finally, $\alpha(b_2\vu_2,b_3\vec v)\le
\alpha(b_2\vu_2,\vu_1)+\alpha(\vu_1,b_3\vec v)
=\gamma_{12}+\gamma(\vu_1,\vec v)\le 0.06\pi$.
Again, that implies $\alpha(b_2\vu_2,b_3\vu_3)<\pi/2$
and the configuration is strictly convex.

Accordingly, from now on assume $\gamma(\vu_1,\vec v)>0.05\pi$.
Since $\gamma_{13}\le\gamma_{23}$, and again
using $A_{i3}=A(\vu_i,\tilde{\vec v})$ for $i=1,2$, it holds
that $\gamma(\vu_1,\vec v)\le\gamma(\vu_2,\vec v)$. Let us
use the fact that the three vectors $\vu_1,\vu_2,\vec v$
lie on the same plane. Accordingly, there are two possibilities
for their ``orientations'' (as usual up to flipping signs).
\Cref{fig:toconv} can be consulted for
illustration.

 \begin{figure}
 \centering
\begin{tikzpicture}
    % the origin
    \coordinate (O) at (0,0);
    \fill[blue!10!white, draw opacity=0.1]   (0,0) -- (96:3) arc(96:90:3) -- cycle;
    \fill[orange!20!white, draw opacity=0.1]   (0,0) -- (90:3) arc(90:84:3) -- cycle;
    \draw[line width=0.2mm,cyan,dashed] (180:4) -- (0:4);
    \draw[line width=0.2mm,cyan,dashed] (90:4) -- (0,0);
    \draw[line width=0.4mm,red,->] (0,0) -- (0:4) node [black,right] {$v$};
    \draw[line width=0.4mm,red,->] (0,0) -- (96:4) node [black,above] {$u_2\;\;$};
    \draw[line width=0.4mm,red,->] (0,0) -- (84:4) node [black,above] {$\;\;u_1$};
    \node[orange!40!black] at (0.2,3.2) {\small $\epsilon$};
    \node[blue] at (-0.17,3.2) {\small $\epsilon$};
    %to clean up he edges the circle and the dot at the origin
    \draw (-4,0) arc(180:0:4);
\end{tikzpicture}
\caption{Illustration of the situation considered in \Cref{claim:toconv}.}
\label{fig:toconv}
\end{figure}

The two possible cases are
$\gamma_{12}=\gamma(\vu_2,\vec v)-\gamma(\vu_1,\vec v)$
and $\gamma_{12}=\pi-\gamma(\vu_2,\vec v)-\gamma(\vu_1,\vec v)$.
In both of them it is easy to check the crucial
property
\begin{equation}\label{eq:08}
\gamma_{12}\le2\left(\frac{\pi}{2}-\gamma(\vu_1,\vec v)\right)\;.
\end{equation}

Applying \Cref{claim:3} to~\eqref{eq:08} a sufficient number of times
(recall that $\gamma(\vu_1,\vec v)\ge 0.05\pi)$,
there exists $K$ such that after agent 1 influences agent 2
for $K$ times, it holds
$\gamma\left(\vu_1,\vu_2^{(K)}\right)\le\big(\pi/2-\gamma(\vu_1,\vec v)\big)/10$.

Finally, we show that $\vu_1,\vu_2^{(K)},\vu_3$ form a strictly
convex configuration. For that, take $b_i=\sign(A_{1i})$
for $i=1,2,3$. 
Since $\gamma_{12},\gamma_{13}<\pi/2$,
and since by \Cref{cor:A_is_non_decreasing} 
it holds that
$\sign(A_{12}^{(K)})=\sign(A_{12})$, it follows
$\bck{\vu_1,b_2\vu_2^{(K)}},\bck{\vu_1,b_3\vu_3}>0$.
It remains to show $\bck{b_2\vu_2^{(K)},b_3\vu_3}>0$.
As discussed above, this is equivalent to
$\bck{b_2\vu_2^{(K)},b_3\vec v}>0$ or in other words
$\alpha\left(b_2\vu_2^{(K)},b_3\vec v\right)<\pi/2$.

Since $\bck{\vu_1,b_2\vu_2^{(K)}},\bck{\vu_1,b_3\vu_3}>0$,
it follows $\alpha(\vu_1,b_2\vu_2^{(K)})=\gamma(\vu_1,\vu_2^{(K)})$
and $\alpha(\vu_1,b_3\vu_3)=\gamma(\vu_1,\vu_3)$.
As before, the second of those equalities implies
$\alpha(\vu_1,b_3\vec v)=\gamma(\vu_1,\vec v)$. Finally,
we can conclude
\begin{align*}
    \alpha(b_2\vu_2^{(K)},b_3\vec v)
    &\le
    \alpha(b_2\vu_2^{(K)},\vu_1)+\alpha(\vu_1,b_3\vec v)
    =\gamma(\vu_1,\vu_2^{(K)})+\gamma(\vu_1,\vec v)\\
    &\le \frac{1}{10}\left(\frac{\pi}{2}-\gamma(\vu_1,\vec v)\right)
    +\gamma(\vu_1,\vec v)
    <\frac{\pi}{2}\;.\qedhere
\end{align*}
\end{proof}

Let us conclude the proof of \Cref{thm:three-ops}.
\begin{proof}[Proof of \Cref{thm:three-ops}]
Since the initial configuration 
$\mathcal{U}^{(0)}$ is not separable, it holds $P^{(0)}<3\pi/2$.
Applying \Cref{claim:52} for $\delta=3\pi/2-P^{(0)}$,
we get $k$ such that, as long as $P^{(t)}>1.2\pi$, then
there is a sequence of $k$ interactions such that
$P^{(t+k)}\le P^{(t)}-0.05\pi$. Iterating for at most
$6$ times, we get a sequence of $K$
interactions such that $P^{(t+K)}\le 1.2\pi$.
So, at every time $t$, independently of the past,
with probability at least $p_{\min}^K$ it will happen that
$P^{(t+K)}\le 1.2\pi$. That implies that
almost surely $P^{(t)}\le 1.2\pi$ at some time $t$.
Furthermore, by \Cref{lem:non-separable-preserved}, the configuration
remains not separable at all times.

Finally, once $P^{(t)}\le 1.2\pi$, by \Cref{claim:toconv}
at every time step the configuration can become strictly convex
in constant time and constant probability.
Again, this implies that $\mathcal{U}^{(t)}$
eventually becomes strictly convex. And
a strictly convex configuration polarizes
by \Cref{lem:strictly-convex-intro}.
\end{proof}

\subsection{Opinions cannot remain $\epsilon$-active forever}
\label{sec:more_than_three_opinions}
When dealing with more than three opinions, parts of our proof strategy remain the same. We will show that almost surely the configuration becomes strictly convex, and deduce from \Cref{lem:strictly-convex-intro} that the opinions polarize.
As a first step towards showing that strict convexity is inevitable, we prove that for every \(\epsilon\), the opinions become \(\epsilon\)-inactive almost surely, as implied by \Cref{lem:almost_sure_inactivity} below.
This section is dedicated to proving that lemma.
\begin{lemma}[Almost sure inactivity]
    \label{lem:almost_sure_inactivity}
    Let $f$ be stable and \(n, d \geq 2\).
    For every $\epsilon>0$, there exists $K_0$ such that for every
    configuration \(\mathcal{U}^{(0)}\),
    there exists a sequence of $K_0$
    interactions such that $\mathcal{U}^{(K_0)}$
    is $\eps$-inactive.
\end{lemma}

\begin{proof}
    We proceed by a natural
    greedy strategy where agents influence each other
    until all opinions are either very correlated or
    almost orthogonal. First, note the following claim.
    \begin{claim}
        \label{lem:close_in_constant_number_of_steps}
        For every \(\epsilon >0\), there exists \(T\) such that
        for every opinions \(\vu_i, \vu_j\),
        if \(\abs{A_{ij}^{(0)}} \ge\epsilon\) and if the first \(T\) interactions are between agents \(i\) and \(j\), then \(\abs{A_{ij}^{(T)}} \geq 1 - \epsilon\).
    \end{claim}
    \begin{proof}
    The proof follows by applying
    \Cref{claim:contraction_stable} a sufficient number of times.
    \end{proof}

    Let \(\epsilon>0\) and set \(\epsilon_0 := \min( 1/8, \epsilon / 64)^2\) so that \(\epsilon \ge 64 \sqrt{\epsilon_0}\).
    We want to find \(K\) and a sequence of \(K\) interactions after which 
    for every pair of opinions, it holds either
    $\abs{A_{ij}^{(K)}} \leq \epsilon$ or
    $\abs{A_{ij}^{(K)}} \geq 1 - \epsilon$.

    Choose (and rename) greedily a maximal set of
    opinions $\vu_1,\ldots,\vu_k$ such that
    $\abs{A_{ij}^{(0)}}\le \eps_0$ for every $1\le i<j\le k$.
    For every opinion $\vu_j$ for $j>k$, by construction
    there exists at least one $1\le i\le k$ such
    that $\abs{A_{ij}^{(0)}}>\eps_0$. For every $j>k$,
    choose one such $i$ arbitrarily and let agent $j$ be influenced at least $T$ times in a row by agent $1$,
    where $T$ is taken from \Cref{lem:close_in_constant_number_of_steps}
    applied for $\eps_0$.
    Perform this for every agent in sequence so that there
    are $K_0=nT$ interactions in total. Note that 
    $K_0$ indeed depends only on $f,n$, and $\epsilon$.

    We now show that the resulting configuration
    is $\eps$-inactive.
    Over the \(K\) interactions, the opinions \(\vu_1, \hdots, \vu_k\) are not moved, so \(\vu_i^{(0)} = \vu_i^{(K)}\) for \(1 \leq i \leq k\). Furthermore, for every \(j > k\), the opinion \(\vu_j\) is influenced by a unique \(\vu_i\) where \(1 \leq i \leq k\). Therefore, 
    from \Cref{lem:close_in_constant_number_of_steps}, for every \(j > k\) there exists \(1 \leq i \leq k\) such that \(\abs{A_{ij}^{(K)}} \geq 1 - \epsilon_0\). It can be observed that:
    \begin{itemize}
        \item By the choice of \(\vu_1, \hdots, \vu_k\), for every \( 1 \leq i, j \leq k\;,\; \abs{A_{ij}^{(K)}} \leq \epsilon_0 < \epsilon\).
        \item For every \( 1 \leq i \leq k\) and \(j > k\), there exists \(1 \leq l \leq k\) such that \(\abs{A_{jl}^{(K)}} \geq 1 - \epsilon_0\). So, either \(i = l\) and \(\abs{A_{ij}^{(K)}} \geq 1 - \epsilon_0 > 1 - \epsilon\); or \(i \neq l\) and with \(\abs{A_{il}^{(K)}} \leq \epsilon_0\le\sqrt{\eps_0}\) and \(\abs{A_{jl}^{(K)}} \geq 1 - \epsilon_0\) we apply \Cref{lem:transitivity_of_inactivity} and obtain 
        \(\abs{A_{ij}^{(K)}} \leq 8 \sqrt{\epsilon_0} < \epsilon\).
        \item For every \( i,j > k\), there exists \(1 \leq l,m \leq k\) such that \(\abs{A_{il}^{(K)}} \geq 1 - \epsilon_0\) and \(\abs{A_{jm}^{(K)}} \geq 1 - \epsilon_0\).
        If \(l = m\), then from \Cref{lem:transitivity_of_closeness}, \(\abs{A_{ij}^{(K)}} \geq 1 - 4\epsilon_0 > 1 - \epsilon\). If \(l \neq m\), then note that \(\abs{A_{lm}^{(K)}} \leq \epsilon_0\) and consider \Cref{cor:transitivity_with_four_opinions} below which is proved by applying \Cref{lem:transitivity_of_inactivity} twice.
        \begin{claim}
        \label{cor:transitivity_with_four_opinions}
            Let \(d \geq 2\), \(\epsilon\ge 0\), and consider four opinions \(\vu_i, \vu_j, \vu_k, \vu_l\).
            If \(\abs{A_{ij}} \leq \sqrt{\eps}\), \(\abs{A_{ik}} \geq 1 - \epsilon\), and \(\abs{A_{jl}} \geq 1 - \epsilon\), 
            then \(\abs{A_{kl}} \le 64 \sqrt{\epsilon}\).
        \end{claim}
        Applying \Cref{cor:transitivity_with_four_opinions} gives \(\abs{A_{ij}^{(K)}} \leq 64 \sqrt{\epsilon_0} \le \epsilon\).
    \end{itemize}
    
        We conclude that for every \(i,j \in [n]\), \(\abs{A_{ij}^{(K)}} \leq \epsilon\) or \(\abs{A_{ij}^{(K)}} \geq 1 - \epsilon\). Therefore the configuration \(\cU^{(K)}\) is \(\epsilon\)-inactive.
\end{proof}

\section{Polarization in two dimensions: Proof of \Cref{thm:2d-main}}
\label{sec:2d}

The objective of this section is to prove
\Cref{thm:2d-main}.
Let $d=2$, \(f\) be a stable function, \(\mathcal{D}\) a distribution with full support,
and \(\mathcal{U}^{(0)}\) a configuration that is not separable.
Let us first point out that
\Cref{thm:2d-main} follows
from \Cref{lem:opinions_in_a_quadrant}:

\begin{proof}[Proof of \Cref{thm:2d-main}]
From \Cref{lem:non-separable-preserved}, since $\cU^{(0)}$
is not separable, $\cU^{(t)}$
remains not separable for every $t\ge 0$.
By \Cref{lem:opinions_in_a_quadrant},
there exists fixed $K$ such that for every $t$, independently of the past,
$\mathcal{U}^{(t+K)}$ is strictly
convex with probability at least $p_{\min}^{K}>0$, where $p_{\min}$ is
the minimum probability in the interaction distribution $\mathcal{D}$.
It follows that almost surely
there exists $t$ such that $\mathcal{U}^{(t)}$
is strictly convex.
And, from \Cref{lem:strictly-convex-intro}, we conclude that almost surely the sequence \((\cU^{(t)})_t\) polarizes.
\end{proof}

It remains to prove~\Cref{lem:opinions_in_a_quadrant}.
Let $\eps_0=1/256$.
By \Cref{lem:almost_sure_inactivity}, for some  $K_0=K_0(n,d,f)$, 
there is sequence of $K_0$
interactions such that $\mathcal{U}^{(K_0)}$ is $\eps_0$-inactive. Therefore, it is sufficient
to show \Cref{lem:opinions_in_a_quadrant}
only for $\eps_0$-inactive, separable configurations.
(Since for a general configuration we can concatenate
$K_0$ interactions that make it $\eps_0$-inactive
and further $K$ interactions that make it strictly convex.)

For a high level idea of the proof in this case,
let $\eps>0$ and consider a configuration
$\vu_1=(\cos(\pi/2-\eps),\sin(\pi/2-\eps)),
\vu_2=(\cos(\pi/2+\eps),\sin(\pi/2+\eps)),
\vu_3=(1,0)$,
as illustrated in \Cref{fig:toconv} (where \(\vec v\) is \(\vu_3\) in our case).
That is, $\vu_1$ and $\vu_2$ are both almost orthogonal
to $\vu_3$ and on the opposite sides of the $y$-axis.
Such a configuration is not strictly convex.
As $\eps$ decreases, any time agent 1 influences agent 2,
opinion $\vu_2$ will move only by a small amount.
However, it requires only a constant number of interactions
to bring $\vu_2$ onto the other side of the $y$-axis,
regardless of the value of $\eps$. Once this is achieved,
the configuration is strictly convex.
While care is needed to get the details right
in every case, at heart the proof of \Cref{lem:opinions_in_a_quadrant} implements this idea.

\begin{proof}[Proof of \Cref{lem:opinions_in_a_quadrant}]
    For readability, let us drop superscripts at time 0.
    That is, we will write \(\mathcal{U},\vu_i, A_{ij}, \gamma_{ij}\), etc.
    instead of \(\mathcal{U}^{(0)},\vu_i^{(0)}, A_{ij}^{(0)}, \gamma_{ij}^{(0)}\). 
    We prove this lemma by cases depending on how the initial configuration $\mathcal{U}$ looks like.
    As discussed above, we can assume w.l.o.g.~that
    $\cU$ is $\eps_0$-inactive for $\eps_0=\frac{1}{256}$.
    Accordingly, 
    for every \(i, j\) either \(\abs{A_{ij}} \leq \epsilon_0\) or \(\abs{A_{ij}} \geq 1 - \epsilon_0\).
    If \(\abs{A_{ij}} \geq 1 - \epsilon_0\) for every $i,j$,
    configuration \(\mathcal{U}\) is already strictly convex by \Cref{claim:convexity-of-one-cluster}.
    
    \begin{figure}
        \caption{The cases from the proof of \cref{lem:opinions_in_a_quadrant}. For Case \(2\), applying  $\vu_1$ and $\vu_2$ to all other opinions in their respective clusters will make the configuration strictly convex.}
        \centering
        \begin{subfigure}[c]{0.4\textwidth}
            \centering
            \caption{\footnotesize Case 1: \(\forall i,j, \, \abs{A_{ij}} \geq 1 - \epsilon_0\). Up to negation all opinions are close to each other.}
            \begin{tikzpicture}
                % the origin
                \coordinate (O) at (0,0);
                % the circle and the dot at the origin
                \draw (O) node[circle,inner sep=1pt,fill] {} circle [radius=\myrad];
                
                % the ``\theta'' arc
                \fill[cyan!10!white, draw opacity=0.1]   (0,0) -- (20:\myrad) arc(20:0:\myrad) -- cycle;
                \fill[cyan!10!white, draw opacity=0.1]   (0,0) -- (180:\myrad) arc(180:200:\myrad) -- cycle;
                \draw[line width=0.2mm,blue,-] (20:\myrad) -- (200:\myrad);
                \draw[line width=0.2mm,blue,-] (180:\myrad) -- (0:\myrad);

                \draw[line width=0.4mm,red,->] (0,0) -- (190:\myrad) node [black,left] {$u_3$};
                \draw[line width=0.4mm,red,->] (0,0) -- (15:\myrad) node [black,right] {$u_1$};
                \draw[line width=0.4mm,red,->] (0,0) -- (5:\myrad) node [black,right] {$u_2$};

                %this is just to make the images the same height
                \draw[line width=0.4mm,white,->] (0,0) -- (250:\myrad) node [white,below] {$u_2$};
                \draw[line width=0.4mm,white,->] (0,0) -- (70:\myrad) node [white,above] {$b_2u_2$};
                \draw[line width=0.4mm,white,->] (0,0) -- (105:\myrad) node [white,above] {$b_4u_4$};

                %to clean up he edges
                \draw (O) node[circle,inner sep=1pt,fill] {} circle [radius=\myrad];
            \end{tikzpicture}
            \label{fig:case1}
        \end{subfigure}
        \hspace{2cm}
        \begin{subfigure}[c]{0.3\textwidth}
            \centering
            \caption{\footnotesize Case 2: \(\exists i,j : \abs{A_{ij}} \leq \epsilon_0\). So, every $b_ku_k$ is within the two blue areas.}
            \begin{tikzpicture}
                % the origin
                \coordinate (O) at (0,0);
                % the circle and the dot at the origin
                \draw (O) node[circle,inner sep=1pt,fill] {} circle [radius=\myrad];
                
                % the ``\theta'' arc
                \fill[cyan!10!white, draw opacity=0.1]   (0,0) -- (70:\myrad) arc(70:110:\myrad) -- cycle;
                \draw[line width=0.2mm,blue,-] (0,0) -- (70:\myrad);
                \draw[line width=0.2mm,blue,-] (0,0) -- (110:\myrad);

                \fill[cyan!10!white, draw opacity=0.1]   (0,0) -- (10:\myrad) arc(10:-10:\myrad) -- cycle ;
                
                \draw[line width=0.2mm,blue,-] (0,0) -- (10:\myrad);
                \draw[line width=0.2mm,blue,-] (0,0) -- (350:\myrad);

                \draw[line width=0.4mm,red,->] (0,0) -- (70:\myrad) node [black,above] {$\qquad b_1\vu_1=\vec v_1$};
                \draw[line width=0.4mm,red,->] (0,0) -- (10:\myrad) node [black,right] {$\vec v_2$};
                \draw[line width=0.4mm,red,->] (0,0) -- (-5:\myrad) node [black,right] {$\vec v_j$};
                \draw[line width=0.4mm,red,->] (0,0) -- (105:\myrad) node [black,above] {$\vec v_i$};

                \draw[line width=0.2mm,cyan,dashed] (180:\myrad) -- (0:\myrad);
                \draw[line width=0.2mm,cyan,dashed] (90:\myrad) -- (270:\myrad);

                %this is just to make the images the same height
                \draw[line width=0.5mm,white,->] (0,0) -- (250:\myrad) node [white,below] {$u_j$};

                %to clean up he edges
                \draw (O) node[circle,inner sep=1pt,fill] {} circle [radius=\myrad];
            \end{tikzpicture}
            \label{fig:case2}
        \end{subfigure}
        \label{fig:allcases}
    \end{figure}

        Therefore, let us turn to the more challenging case
        where there exist $i,j$ such that \(\abs{A_{ij}} \leq \epsilon_0\).
        Recall the notation $\alpha_{ij}=\arccos(\langle \vu_i,\vu_j\rangle)$
        and the
        effective angle 
        $\gamma_{ij}=\min(\alpha_{ij},\pi-\alpha_{ij})$.
        Without loss of generality, assume that \(\vu_1\) and \(\vu_2\) satisfy
        \begin{equation*}
            \abs{A_{12}} = \max_{i,j\,:\,\abs{A_{ij}} \leq \epsilon_0} \abs{A_{ij}} \; .
        \end{equation*}
        Note that this is equivalent to
        \begin{equation}
            \label{eq:min_gamma_12}
            \gamma_{12} = \min_{i,j\,:\,\abs{A_{ij}} \leq \epsilon_0} \gamma_{ij} \; .
        \end{equation}         
        By \Cref{lem:cluster}, the opinions are partitioned into
        two clusters \(\mathcal{C}_1\) and \(\mathcal{C}_2\) with
        \begin{equation*}
            \mathcal{C}_1 = \{i \in [n] : \abs{A_{1i}} \geq 1 - \epsilon_0 \}
            \quad \text{ and } \quad
            \mathcal{C}_2 = \{i \in [n] : \abs{A_{2i}} \geq 1 - \epsilon_0 \} \; ,
        \end{equation*}
        and \(\vu_1\) and \(\vu_2\) are (up to negation) 
        a closest pair of opinions which are not in the same cluster.

        Let us specify a sequence of interactions that will make
        the configuration strictly convex.
        By \Cref{claim:3}, there exists some fixed $T$ such that
        if $\gamma_{ij}\le\pi/4$ and $i$ influences $j$ for $T$ times,
        then $\gamma_{ij}^{(T)}\le \gamma_{ij}/8$.
        Take a sequence of \(K = (n-2) T\) interactions where
        for every $i>2$: if $i\in\mathcal{C}_w$, then the sequence contains $T$ interactions
        where $w$ influences $i$.
        It remains to
        show that after this sequence 
        the configuration \(\mathcal{U}^{(K)}\) is strictly convex.

        For every \(i \in [n]\), define
        \begin{equation*}
            b_i := 
            \begin{cases*}
                \sign{A_{1i}} \; , \text{ if } i \in \mathcal{C}_1  \\
                \sign{(A_{12})}\sign{(A_{2i})} \; , \text{ if } i \in \mathcal{C}_2
            \end{cases*}
            \quad \text{ and } \quad
            \vec v_i^{(t)} := b_i \vu_i^{(t)} \; .
        \end{equation*}
        To show strict convexity at time $K$,
        it is sufficient that for every \(i,j\)
        it holds \(\, \bck{\vec v_i^{(K)}, \vec v_j^{(K)}} > 0\). For that, we consider the angles \[\beta_{ij}^{(t)} = \arccos{\left(\bck{\vec v_i^{(t)}, \vec v_j^{(t)}}\right)}\]
        for $0\le t\le K$. Note that either \(\beta_{ij}^{(t)} = \gamma_{ij}^{(t)}\) or \(\beta_{ij}^{(t)} = \pi - \gamma_{ij}^{(t)}\). The result will follow from the sequence of claims below. 
        
        \begin{claim}\label{cl:6}
            For every \(0 \leq t \leq K\),
            \(\beta_{12}^{(t)} = \beta_{12} = \gamma_{12}\).
        \end{claim}
        This is clear, because no agent influences \(1\) or \(2\) among the \(K\) interactions, so for every \(0 \leq t \leq K\), \(\vu_1^{(t)} = \vu_1\) and \(\vu_2^{(t)} = \vu_2\). Furthermore, by construction
        we have:
        \begin{claim}\label{cl:7}
            Let $w\in\{1,2\}$.
            For every \(i \in \mathcal{C}_w\), \(\gamma_{wi}^{(K)} \le \gamma_{wi}/8\).
        \end{claim}
        \begin{claim}\label{cl:8}
            For every \(t\),
            if \(i \in \mathcal{C}_1\), then \(b_i = \sign{(A_{1i}^{(t)})}\);
            and if \(i \in \mathcal{C}_2\), then \(b_i = \sign{(A_{12}^{(t)})} \sign{(A_{2i}^{(t)})}\).
        \end{claim}
        The statement of \Cref{cl:8} is true at $t=0$
        by definition. It remains true at subsequent times,
        because $A_{12}^{(t)}$ remains the same, and,
        for every \(i \in \mathcal{C}_w\), \(i\) is influenced only by \(w\). Therefore, for every \(0 \leq t < K\), \(\sign{(A_{wi}^{(t+1)})} = \sign{(A_{wi}^{(t)})}\). 
        \begin{claim}\label{cl:9}
            For every \(0 \leq t \leq K\) and \(i \in \mathcal{C}_w\), \(\beta_{wi}^{(t)} = \gamma_{wi}^{(t)}\).
        \end{claim}
        \begin{proof}
        In short, this follows from the previous claim.
        If \(i \in \mathcal{C}_1\), then
        \begin{equation*}
            \bck{\vec v_1^{(t)}, \vec v_i^{(t)}} = b_1 b_i A_{1i}^{(t)} = \sign{(A_{1i}^{(t)})} \cdot A_{1i}^{(t)} \ge 0 \; .
        \end{equation*}
        Thus, \(\beta_{1i}^{(t)} \le \pi/2\) and \(\beta_{1i}^{(t)} = \gamma_{1i}^{(t)}\).
        Similarly, if \(j \in \mathcal{C}_2\), then 
        \begin{equation*}
            \bck{\vec v_2^{(t)}, \vec v_j^{(t)}} = b_2 b_j A_{2j}^{(t)} = 
            b_2 \sign(A_{12}) \sign(A_{2j}^{(t)}) \cdot A_{2j}^{(t)} \ge 0 \; .
        \end{equation*}
        Thus, \(\beta_{2j}^{(t)} \le\frac{\pi}{2}\) and \(\beta_{2j}^{(t)} = \gamma_{2j}^{(t)}\).
        \end{proof}
        \begin{claim}
            For every \(i,j \in [n]\), \(\beta_{ij}^{(K)} < \pi/2\).
        \end{claim}
        \begin{proof}
        If \(i,j \in \mathcal{C}_w\) for $w=1,2$, then,
        by \Cref{cl:9} and \Cref{cl:7},
        \begin{equation*}
            \beta_{ij}^{(K)}
            \leq \beta_{wi}^{(K)} + \beta_{wj}^{(K)}
            = \gamma_{wi}^{(K)} + \gamma_{wj}^{(K)}
            \leq \frac{\gamma_{wi}}{8} + \frac{\gamma_{wj}}{8}
            \leq \frac{\pi}{8} < \frac{\pi}{2} \; .
        \end{equation*}
        Now suppose \(i \in \mathcal{C}_1\) and \(j \in \mathcal{C}_2\).
        Similarly, and also using \Cref{cl:6}, we have
        \begin{equation*}
            \beta_{ij}^{(K)}
            \leq \beta_{1i}^{(K)} + \beta_{12}^{(K)} + \beta_{2j}^{(K)}
            = \gamma_{1i}^{(K)} + \gamma_{12}^{(K)} + \gamma_{2j}^{(K)}
            \leq  \beta_{12} + \frac{\beta_{1i} + \beta_{2j}}{8} \; .
        \end{equation*}
        Let's show that \(\beta_{2j} / 4 < \pi/2 - \beta_{12}\).
        Suppose for the sake of contradiction that \(\beta_{2j} / 4 \geq \pi/2 - \beta_{12}\).
        See \Cref{fig:case2} for an illustration of the situation.
        
        First, note that we have 
        % \abdou{I am not sure we use $\gamma_{2j}=\beta_{2j}$}
        $\gamma_{2j}=\beta_{2j}\le\arccos(1-\epsilon_0)
        <\arccos(\epsilon_0)\le\beta_{12}=\gamma_{12}\le\pi/2$. In particular,
        given that $d=2$, it holds that 
        \begin{equation*}
            \beta_{1j} = (\beta_{12} + \beta_{2j})
            \quad \text{ or } \quad
            \beta_{1j} = (\beta_{12} - \beta_{2j})
        \end{equation*}
        (since $\beta_{12}+\beta_{2j}<\pi$ and $\beta_{12}-\beta_{2j}>0$ 
        there is no need for taking modulo).
        \begin{itemize}
            \item \(\beta_{1j} = \beta_{12} - \beta_{2j}\) for
            $\beta_{2j}>0$ is impossible, since then
            \(\gamma_{1j} \le \beta_{1j} < \beta_{12} = \gamma_{12}\),
            which contradicts \eqref{eq:min_gamma_12}.
            \item If \(\beta_{1j} = \beta_{12} + \beta_{2j}\), then
            it follows by our assumption by contradiction that
            \(\beta_{1j} \geq 2 \pi - 3 \beta_{12}\) and consequently
            \(\gamma_{1j} \leq \pi - \beta_{1j} \leq 3 \beta_{12} - \pi = 3 \gamma_{12} - \pi < \gamma_{12}\), again contradicting
            \eqref{eq:min_gamma_12}.
            The last inequality
            is strict as $\gamma_{12}=\pi/2$ implies that clusters $\mathcal{C}_1$ and $\mathcal{C}_2$ are orthogonal and
            therefore $\mathcal{U}$ is separable.            
        \end{itemize}
        We conclude that \(\beta_{2j}/4 < \pi/2 - \beta_{12}\), and with a similar argument, \(\beta_{1i}/4 < \pi/2 - \beta_{12}\). These imply that
        $\beta_{ij}^{(K)}\le\beta_{12}+(\beta_{1i}+\beta_{2j})/8
        <\pi/2$.
        This ends the proof of the claim.
        \end{proof}

        We conclude that for every \(i\) and \(j\), \(\beta_{ij}^{(K)} < \frac{\pi}{2}\) which means \(\bck{\vec v_i^{(K)}, \vec v_j^{(K)}} > 0\). Hence, the configuration \(\mathcal{U}^{(K)}\) is strictly convex.
    \end{proof}

\section{Counterexample in higher dimensions:\\Proof of \Cref{thm:counterexample}}
\label{sec:counterexample}

In this section, we prove that the statement of Lemma~\ref{lem:opinions_in_a_quadrant} is false for $d\geq 3$, indicating that a generalized proof of polarization requires further ideas.

\begin{proof}[Proof of \cref{thm:counterexample}]
As before, recall that $A_{ij}^{(t)}:=\langle \vu_i^{(t)}, \vu_j^{(t)} \rangle$.
First, observe that, if three opinions given by $\vu_1^{(t)}$, $\vu_2^{(t)}$, $\vu_3^{(t)}$ are in a strictly convex configuration, then the following holds: If $A_{12}^{(t)}>0$, then we have~sign($A_{13}^{(t)}$)=sign($A_{23}^{(t)}$) 
(because $A_{12}>0$ implies $b_1=b_2$ in the strictly convex definition).

Let us define the configuration satisfying the required conditions.
Pick $\epsilon>0$ small enough such that it satisfies the condition
$\epsilon/2\geq \epsilon^2 \cdot \left(\eta + (2\eta + \eta^2)/2 \right)T (1+\eta)^{2T} $.
Pick three vectors $\vu_1^{(0)}, \allowbreak\vu_2^{(0)}, \vu_3^{(0)}$ in $\mathbb{S}^{d-1}$ such that $A_{12}^{(0)}=A_{13}^{(0)}=\epsilon$, $A_{23}^{(0)}=-\epsilon$ (note that this is possible for small $\eps$ whenever $d \geq 3$).

Since $A_{12}^{(0)}>0$ but $A_{13}^{(0)}, A_{23}^{(0)}$ do not have the same sign, this configuration is not strictly convex (and clearly it is not separable). Furthermore, it will not be strictly convex until one of the signs of the $A_{ij}$ flips. Therefore, it is sufficient to show that none of the signs of the $A_{ij}$ flips in the first $T$ steps. It turns out that
any possible interaction can change each correlation
by at most $O(\eps^2)$. Hence, as $\eps\to 0$, more and
more interactions are required to flip one of the signs.
In the rest of the proof, we develop this argument rigorously.

\medskip

To that end, we investigate the changes in correlations for every possible interaction. For this, let $(i,j,k)$ be some permutation of $(1,2,3)$ and assume that, at time $t$, agent~$i$ is influenced by agent~$j$. We have the following upper bounds on the changes in correlations:
\begin{equation*}
\left|A_{ij}^{(t+1)}\right| 
\overset{\eqref{eq:new-correlation}}{=}
\frac{(1+\eta)\cdot \left|A_{ij}^{(t)}\right|}{\sqrt{1+(2 \eta + \eta^2)\left(A_{ij}^{(t)}\right)^2}}
\leq (1+\eta) \cdot \left|A_{ij}^{(t)}\right| 
\leq (1+\eta) \cdot \max \left\{\left|A_{ij}^{(t)}\right| , \left|A_{ik}^{(t)}\right| , \left|A_{jk}^{(t)}\right| \right\},
\end{equation*}
and
\begin{align*}
\left|A_{ik}^{(t+1)}\right|
{=}\frac{\left| A_{ik}^{(t)}+\eta A_{ij}^{(t)}A_{jk}^{(t)} \right|}{\sqrt{1+(2 \eta + \eta^2)\left(A_{ij}^{(t)}\right)^2}}
&\leq \left|A_{ik}^{(t)}\right|+\eta \left|A_{ij}^{(t)}\right| \cdot \left|A_{jk}^{(t)}\right| \; , \\
&\leq (1+\eta) \cdot \max \left\{\left|A_{ij}^{(t)}\right| , \left|A_{ik}^{(t)}\right| , \left|A_{jk}^{(t)}\right| \right\}. 
\end{align*}
Together, this gives for every $t \geq 0$
\begin{equation}\label{eq:correlation-lower-bound}
\max \left\{\left|A_{ij}^{(t)}\right| , \left|A_{ik}^{(t)}\right| , \left|A_{jk}^{(t)}\right| \right\}
 \leq (1+\eta)^t \cdot \epsilon.
\end{equation}

To derive lower bounds on the changes in correlations, the following inequality will turn out to be useful. For every $x\geq 0$, Taylor's theorem implies that there is some $\xi \in [0,x]$ such that
\begin{equation}\label{eq:Taylor-sqrt}
    \frac{1}{\sqrt{1+x}}=1 - \frac{x}{2} + \frac{3}{4} (1+\xi)^{-5/2} \geq 1 - \frac{x}{2}.
\end{equation}
We obtain the following bounds on the changes in correlation.
By \Cref{cor:A_is_non_decreasing},
$A_{ij}^{(t+1)}$ has the same sign as $A_{ij}^{(t)}$ and
$|A_{ij}^{(t+1)}|\ge |A_{ij}^{(t)}|$.
Moreover, we have
\begin{align*}
A_{ik}^{(t+1)} \cdot \mathrm{sign}\left(A_{ik}^{(t)}\right)
{=}&
\frac{A_{ik}^{(t)}+\eta A_{ij}^{(t)}A_{jk}^{(t)}}{\sqrt{1+(2 \eta + \eta^2)\left(A_{ij}^{(t)}\right)^2}} \cdot \mathrm{sign}\left(A_{ik}^{(t)}\right)
\geq \frac{\left| A_{ik}^{(t)}\right|-\eta \left| A_{ij}^{(t)}A_{jk}^{(t)}\right|}{\sqrt{1+(2 \eta + \eta^2)\left(A_{ij}^{(t)}\right)^2}}\\
\overset{\eqref{eq:Taylor-sqrt}}{\geq}& \left( \left| A_{ik}^{(t)}\right| -\eta \cdot \left| A_{ij}^{(t)}\right|  \cdot \left| A_{jk}^{(t)}\right| \right) \cdot \left(1- \frac{2\eta + \eta^2}{2}\left(A_{ij}^{(t)}\right)^2 \right)\\
&\geq \left| A_{ik}^{(t)}\right| - \eta \cdot \left| A_{ij}^{(t)}\right| \cdot \left| A_{jk}^{(t)}\right|-\frac{2\eta + \eta^2}{2}\left(A_{ij}^{(t)}\right)^2 \\
\overset{\eqref{eq:correlation-lower-bound}}{\geq}& \left| A_{ik}^{(t)}\right|-\left( \eta + \frac{2\eta + \eta^2}{2} \right) (1+\eta)^{2t} \cdot \epsilon^2.
\end{align*}
Considering two cases $A_{ik}^{(t)}>0$ and $A_{ik}^{(t)}<0$,
the calculation above implies that  $A_{ik}^{(t+1)}$ has the same sign as $A_{ik}^{(t)}$ as long as the right hand side is strictly positive.
Moreover, the absolute value of the
correlation decreases by at most $\left( \eta + (2\eta + \eta^2)/2 \right) (1+\eta)^{2t} \cdot \epsilon^2.$
Together, we obtain that, for every $i \neq j$ in $\{1,2,3\}$ and every $t \in \{0, \dots, T\}$, we have
\begin{align*}
\left|A_{ij}^{(t)}-A_{ij}^{(0)}\right|&\le \epsilon^2 \cdot \left( \eta + \frac{2\eta + \eta^2}{2} \right) \cdot \sum\limits_{k=0}^{T-1} (1+\eta)^{2k}
\le \epsilon^2 \cdot \left( \eta + \frac{2\eta + \eta^2}{2} \right) \cdot T\cdot (1+\eta)^{2T}
\le \frac{\epsilon}{2}\;.
\end{align*} 
Since $\left|A_{ij}^{(0)}\right|=\epsilon$,
this implies that the sign of $A_{ij}$ does not flip for the first $T$ steps. In particular, the configuration does not become strictly convex for any sequence of $T$ interactions.
\end{proof}

\section{Polarization for \((d,n,\mathcal{D})\)-stable functions:\\Proof of~\Cref{thm:dn-stable}}
\label{sec:3d-stable}

Our strategy to prove \Cref{thm:dn-stable} is as follows: As a first step, we
prove \Cref{lem:clusters_preserved},
which implies
that the partition into clusters does not change while a configuration is $\epsilon$-inactive.
Recall from the definition of $(d,n,\mathcal{D})$-stable functions
that an $\epsilon$-inactive configuration with at least two clusters 
almost surely becomes $\eps$-active again.
Accordingly, in the second step, we prove that, once a configuration becomes $\epsilon$-active, there exists a constant length sequence of interactions that makes the configuration $\epsilon$-inactive again, but with a strictly
smaller number of clusters. In particular, the number of clusters decreases with some constant probability.

More precisely, we only give such a constant length sequence
if the configuration became active due to $|A_{ij}^{(t)}|>\eps$
for $i,j$ from different clusters
(as opposed to $|A_{ij}^{(t)}|<1-\eps$ for $i,j$ from the same cluster).
Since this indeed occurs with constant probability by the definition of $(d,n,\mathcal{D})$-stability, 
it can also be concluded that the
number of clusters decreases with constant probability.

Putting these facts together,
if a configuration is $\eps$-inactive with $k>1$ clusters, with constant probability
it will become $\eps$-active and
then $\eps$-inactive again with a strictly
smaller number of clusters.
As the number of clusters is bounded
by $\max(d, n)$, almost surely the configuration ends up with only one cluster.
But a configuration with one cluster must be strictly convex, and therefore polarizes
by \Cref{lem:strictly-convex-intro}.

\medskip

Let us proceed with this plan,
starting with the proof of \Cref{lem:clusters_preserved}.
As a preliminary, we give a technical bound which formalizes that, for stable functions, interactions between nearly orthogonal opinions result in only small changes to the opinions.

\begin{claim}
\label{lem:small_old_new_correlation}
    Let \(f\) be a stable update function and $M,c>0$ such that $|f(A)|\le M$ if  $|A|\le c$.
    Then, whenever agent $j$ influences agent $i$ and $|A_{ij}|\le c$, it holds for the
    new opinion $\vu'_i$ that
    \(A(\vu_i', \vu_i) \geq 1 - M\).
\end{claim}
\begin{proof}
    From \(\sign(f(A)) = \sign(A)\), we have \(A f(A) \geq 0\). Suppose $A=A_{ij}$
    satisfies \(\abs{A} \leq c\). Then,
    \begin{equation*}
        A(\vu_i', \vu_i)
        = \frac{1 + A f(A)}{\sqrt{1 + 2 A f(A) + f(A)^2}}
        \geq \frac{1}{1 + \abs{f(A)}}
        \geq \frac{1}{1 + M}
        = 1 - \frac{M}{1 + M}
        \geq 1 - M \; . \qedhere
    \end{equation*}
\end{proof}

\begin{proof}[Proof of \Cref{lem:clusters_preserved}]
    Let \(\epsilon_1 = 1/256\).
    Given that \(f\) is continuous on \([-1, 1]\) and \(f(0) = 0\), it follows that
    for every \(M > 0\) there exists \(c(M) > 0\) such that \(\abs{f(A)} \leq M\) whenever \(\abs{A} \leq c(M)\).
    Take \(M = \left(\epsilon_1/8\right)^2\) 
    and choose \(\epsilon_0 = \min(c(M), M)\).
    Suppose \(\cU^{(t)}\) is \(\epsilon_0\)-inactive.
    Assume that agent \(j\) influences agent \(i\), so only \(\vu_i\) moves from time \(t\) to \(t+1\).
    We want to show that \(\cU^{(t+1)}\) is \(\epsilon_1\)-inactive and is formed by the same clusters as those that compose \(\cU^{(t)}\). We proceed by analysing two cases, depending on
    if $i$ and $j$ were in the same cluster or not.

    Assume that \(i\) and \(j\) were in the same cluster at time \(t\).
    Their interaction does not decrease their absolute correlation,
    that is \(\abs{A_{ij}^{(t+1)}} \geq \abs{A_{ij}^{(t)}} \geq 1 - \epsilon_0\).
    So, for any \(k\) in the same cluster, and from the transitivity of closeness stated in \Cref{lem:transitivity_of_closeness}, we have 
    \(\abs{A_{ik}^{(t+1)}} \geq 1 - 4 \epsilon_0\ge 1-\eps_1>1/2\) because \(\abs{A_{jk}^{(t+1)}} = \abs{A_{jk}^{(t)}} \geq 1 - \epsilon_0\).
    Also, for every \(k\) from another cluster, we have \(\abs{A_{jk}^{(t+1)}} \leq \epsilon_0\). Therefore, \(\abs{A_{ik}^{(t+1)}} \leq 8 \sqrt{\epsilon_0}\le \eps_1<1/2\) from \Cref{lem:transitivity_of_inactivity}.
    Therefore, \(\cU^{(t+1)}\) is \(\epsilon_1\)-inactive 
    with the same clusters as $\cU^{(t)}$.
    
    Now suppose \(i \in C_1\) and \(j \in C_2\) where \(C_1\) and \(C_2\) are different clusters of \(\cU^{(t)}\). 
    As $\eps_0\le c(M)$, from \Cref{lem:small_old_new_correlation}, this implies \(\abs{A(\vu_i^{(t)}, \vu_i^{(t+1)})} \geq 1 - M\).
    For every \(k \in C_1\), \(\abs{A_{ik}^{(t)}} \geq 1 - \epsilon_0\ge 1-M\). Therefore \(\abs{A_{ik}^{(t+1)}} \geq 1 - 4 M\ge 1-\eps_1>1/2\) by \Cref{lem:transitivity_of_closeness}.
    For every \(k \not\in C_1\), \(\abs{A_{ik}^{(t)}} \leq \epsilon_0\le\sqrt{M}\). Therefore \(\abs{A_{ik}^{(t+1)}} \leq 8 \sqrt{M}=\eps_1<1/2\) by \Cref{lem:transitivity_of_inactivity}.
    Again, we conclude that $\cU^{(t+1)}$ is $\eps_1$-inactive
    with unchanged clusters.
\end{proof}

For an $\eps$-inactive configuration, let NC\((\mathcal{U})\) denote the number of clusters in~\(\mathcal{U}\).
In the following lemma, we argue that if a configuration becomes active, then there exists a bounded number of interactions
that strictly decreases the number of clusters.

\begin{lemma}
\label{lem:number_of_clusters_decrease}
    Let $d, n \geq 2$, $f$ be a stable update function and 
    take $\eps_0$ from \Cref{lem:clusters_preserved}.
    For every
    $0<\eps\le \eps_0$ there exists $K_0 = K_0(n, f, \epsilon)$
    such that:
    
    If \(\mathcal{U}^{(t-1)}\) is \(\epsilon\)-inactive
    and there exist two agents \(i, j\) that were in different clusters at time \(t-1\) and such that 
    \(\abs{A_{ij}^{(t)}} > \epsilon\),
    then, there exists a sequence of \(K \leq K_0\) interactions such that \(\mathcal{U}^{(t+K)}\) is \(\epsilon\)-inactive and \(\NC\left(\mathcal{U}^{(t+K)}\right) \leq \NC\left(\mathcal{U}^{(t-1)}\right) - 1\).
\end{lemma}
\begin{proof}
    Suppose \(\cU^{(t-1)}\) is \(\epsilon\)-inactive for
    $0<\eps\le \eps_0$.
    From \Cref{lem:clusters_preserved}, \(\mathcal{U}^{(t)}\) is \(1/256\)-inactive, and the clusters are preserved from \(t-1\) to \(t\).
    Now suppose w.l.o.g.~there are two clusters \(C_1\) and \(C_2\) and two agents \(1 \in C_1\), \(2 \in C_2\) such that \(\epsilon < \abs{A_{12}^{(t)}}\le 1/256\).
    Let's find a sequence of \(K > 0\) interactions such that \(\mathcal{U}^{(t+K)}\) is \(\epsilon\)-inactive and NC\((\mathcal{U}^{(t+K)}) \leq \text{NC}(\mathcal{U}^{(t-1)}) - 1\).
    
    Let \(\epsilon' = \epsilon^2 / 16\).
    From \Cref{lem:close_in_constant_number_of_steps}, there exists \(T = T(\epsilon')\) such that for any pair of agents with 
    $\abs{A_{ij}^{(t)}}\ge\epsilon'$, after $T$ interactions between
    them it holds $\abs{A_{ij}^{(t+T)}}\ge 1-\epsilon'$.
        
    In particular, for every \(i \in C_1\),
    it holds \(\abs{A_{1i}^{(t)}} \geq 1 - 1/256>\epsilon'\). 
    If agent \(1\) influences agent \(i\) for \(T\) successive steps, we will obtain \(\abs{A_{1i}^{(t+T)}}\ge 1 - \epsilon'\). We can repeat this for every \(i \in C_1\). Therefore, there exists \(T_1 > 0\) and a sequence of \(T_1\) interactions such that at time \(t+T_1\) we have \(\abs{A_{1i}} \geq 1 - \epsilon'\) for every \(i \in C_1\).

    On the other hand, let $j\in C_2$, in particular \(\abs{A_{2j}^{(t)}} \geq 1 - 1/256\).
    Similarly as before, let agent 2 influence agent $j$ for $T(\eps')$ times,
    after which
    we will have $|A_{2j}|\ge 1-\eps'$.
    At that time, from the transitivity of activity stated in \Cref{lem:transitivity_of_activity} (and by noting that \(\epsilon^2 = 16 \epsilon'\)), we also have \(\abs{A_{1j}}\ge \sqrt{\epsilon'} > \epsilon'\).
    Accordingly, we now let agent \(1\) influence agent \(j\) further $T(\eps')$ times, which makes agents \(1\) and \(j\) close, in particular $|A_{1j}|\ge 1-\eps'$.
    
    Repeating this procedure for every $j\in C_2$, 
    there exists \(T_2 > 0\) and a sequence of \(T_2\) interactions such that at time \(t+T_1+T_2\) we have
    \(\abs{A_{1i}} \geq 1 - \epsilon'\) for every \(i \in C_1\cup C_2\).
    Let \(K_1 = T_1 + T_2\). We conclude that for every \(i \in C_1\cup C_2\), we have 
    \(\abs{A_{1i}^{(t+K_1)}} \geq 1 - \epsilon'\). From Lemma \ref{lem:transitivity_of_closeness}, we have \(\abs{A_{ij}^{(t+K_1)}} \geq 1 - 4 \epsilon' = 1 - \epsilon^2/4 \geq 1 - \epsilon\), for every \(i, j \in C_1 \cup C_2\).

    If the configuration \(\mathcal{U}^{(t+K_1)}\) is $\eps$-inactive, we are done. Indeed, all
    opinions from $C_1\cup C_2$ are now in one cluster, and opinions in other clusters did not move,
    so   
    we must have reduced the number of clusters by (at least) one.
    
    On the other hand, if $\mathcal{U}^{(t+K_1)}$ is \(\epsilon\)-active, recall that $\mathcal{U}^{(t-1)}$ was $\eps$-inactive.
    Since then, all interactions moved only opinions belonging to $C_1$
    or $C_2$. Therefore, if $\cU^{(t+K_1)}$ is $\eps$-active, it must be
    because there is~$i_1\in C_1\cup C_2,i_\ell\in C_\ell$ such that $|A_{i_1i_\ell}^{(t+K_1)}|>\eps$
    and $\ell\ge 3$.

    But now we can repeat the same argument with respect to clusters
    $C_1\cup C_2$ and $C_\ell$. (In particular, note that at time $t+K_1$ we have
    $|A_{i_1 i}|\ge 1-\eps$ for every $i\in C_1\cup C_2$ and
    $|A_{i_\ell i}|\ge 1-\eps$ for every $i$ in $C_\ell$.)
    Therefore, there exists $K_2$ and a sequence of interactions
    (all of them moving opinions only in $C_1,C_2$ and $C_\ell$) such that at time $t+K_1+K_2$ we have $|A_{ij}^{(t+K_1+K_2)}|\ge 1-\eps$ for every $i,j\in C_1\cup C_2\cup
    C_\ell$.

    Repeating this procedure at most $n$ times if necessary, we finally arrive
    at a sequence of $K$ interactions such that $\cU^{(t+K)}$ is $\eps$-inactive.
    By construction, it is clear that this configuration has
    a strictly smaller number of clusters than $\cU^{(t-1)}$. In particular,
    at least two clusters of $\cU^{(t-1)}$ have been merged into one.
    By noting that for every \(i\), \(K_i \leq 3 n T(\epsilon')\), and \(i \leq n\), we can choose \(K_0 = 3 n^2 T(\epsilon')\), and the claim is verified.
\end{proof}

We are now ready to prove \Cref{thm:dn-stable}.
Note that the preceding results required only
that $f$ is stable. Only this proof
uses the stronger assumption that $f$ is $(d,n,\mathcal{D})$-stable.

\begin{proof}[Proof of \Cref{thm:dn-stable}.]
    Recall \(\epsilon\) from the definition of $(d,n,\mathcal{D})$-stability.
    Let $\cU^{(0)}$ be an initial configuration which is not
    separable (and hence remains not separable for the
    rest of time).
    If it is \(\epsilon\)-active, then from \Cref{lem:almost_sure_inactivity}, a.s.~there exists \(t_1\) such that \(\cU^{(t_1)}\) is \(\epsilon\)-inactive.
    Therefore, let us assume w.l.o.g.~that
    \(\mathcal{U}^{(0)}\) is also \(\epsilon\)-inactive. 
    
    Let $t(0)=0$ and $\NC_0$ be the number of clusters of configuration
    $\mathcal{U}^{(0)}$. If $\NC_0=1$, then by \Cref{claim:convexity-of-one-cluster} the configuration
    is strictly convex and from \Cref{lem:strictly-convex-intro} \(\cU^{(t)}\) polarizes almost surely.

    Suppose $\NC_0>1$. The function \(f\) is \((d,n,\mathcal{D})\)-stable.
    Therefore, almost surely there exists the earliest \(t\) such that \(\mathcal{U}^{(t)}\) is \(\epsilon\)-active.
    From \Cref{lem:almost_sure_inactivity}, almost surely there exists \(t(1) > t\) such that
    the configuration \(\mathcal{U}^{(t(1))}\) first returns to being \(\epsilon\)-inactive.
    
    Repeating this argument, we let
    \(t(k)\) as the first \(t > t(k-1)\) such that \(\cU^{(t-1)}\) is \(\epsilon\)-active and \(\cU^{(t)}\) is \(\epsilon\)-inactive. Then, we let $\NC_k$ to be the number of clusters of the configuration at time $t(k)$.    
    The sequence $t(k)$ is either infinite or terminates at the first $k$ such that
    $\NC_k=1$. In the latter case, we extend the sequence $(\NC_k)$ so that $\NC_s=1$ for every
    $s>k$.

    As mentioned, if $k$ is the first index such that $\NC_k=1$, then the configuration
    at time $t(k)$ is strictly convex and $\mathcal{U}^{(t)}$ polarizes. Accordingly,
    we will be done if we prove that almost surely $\NC_k=1$ for some~$k$.

    To that end, condition on some history of configurations up to time $t(k)$ such that
    $\NC_k>1$. Let $t'>t(k)$ be the first later time such that
    the configuration becomes $\eps$-active.
    By \Cref{lem:clusters_preserved}, the configurations at times $t'-1$ and $t'$
    have the same clusters as at time $t(k)$. Furthermore, by the definition
    of $(d,n,\mathcal{D})$-stability, with probability at least $p>0$, at time $t'$ there exist
    two agents from different clusters with correlation exceeding $\eps$. By
    \Cref{lem:number_of_clusters_decrease}, in that case there exists a sequence of at most
    $K_0$ interactions such that afterwards the configuration becomes
    $\eps$-inactive with a strictly smaller number of clusters. In other words,
    after those interactions it will hold $\NC_{k+1}<\NC_k$.

    Let $w>1$.
    Averaging over all histories such that $\NC_k=w$, we conclude that
    $\Pr[\NC_{k+1}\le w-1\;\vert\;\NC_k=w]\ge p\cdot p_{\min}^{K_0}$,
    where $p_{\min}>0$ is the smallest probability in the interaction
    distribution $\mathcal{D}$.
    
    To show that almost surely there exists \(k\) such that \(\text{NC}_k = 1\), we couple the random process \(\text{NC}_k\) with the following Markov chain.
    Let \(X_k \in \{1,\ldots,n\}\) for \(k \in \mathbb{N}\). Let $q=p\cdot p_{\min}^{K_0}$. We define
    \begin{align*}
        &X_0 = n \; , \\
        &\mathbb{P}\left[X_{k+1} = 1 | X_k = 1 \right] = 1 \; , \\
        &\forall m > 1, \; 
        \begin{cases*}
            \mathbb{P}\left[X_{k+1} = m-1 \, | \, X_k = m \right] = q \\
            \mathbb{P}\left[X_{k+1} = n \, | \, X_k = m \right] = 1 - q
        \end{cases*} \; .
    \end{align*}
    Let's show that there exists a joint distribution \((\text{NC}_k, X_k)\) such that almost surely \(X_k \geq \mathrm{NC}_k\) for every $k$.
    We proceed by induction on $k$.
    For $k=0$ indeed it holds with certainty
    $\NC_0\le n=X_0$.
    
    Now suppose inductively that the joint distribution \((\text{NC}_{s}, X_{s})_{s\le k}\) is already defined 
    up to some $k$.  
    Condition on some \(\text{NC}_k = w\) and \(X_k = m\).
    By induction, \(X_k \geq \text{NC}_k\) almost surely, in particular $w\le m$.

    If $w=m=1$, by definition we have almost surely $X_{k+1}=1=\NC_{k+1}$.

    Conversely, let $m>1$.
    By the definition of $X_k$, conditioned on $X_k=m$, we need
    \(X_{k+1} = m - 1\) with probability \(q\) and \(X_{k+1} = n\) with probability \(1-q\).
    But as stated above, $\Pr[\NC_{k+1}\le w-1\;\vert\;\NC_k=w]\ge q$, so indeed there
    exists a joint distribution such that if $X_{k+1}=m-1$, then $\NC_{k+1}\le w-1\le m-1$
    almost surely. Otherwise, it holds $\NC_{k+1}\le n=X_{k+1}$.    
    We conclude that almost surely for every \(k\), \(\text{NC}_{k} \leq X_{k}\).

    Since $1$ is the only recurrent state in the Markov chain $(X_k)$, almost surely
    there exists $k$ such that $X_k=1$.
    (This is because the Markov chain has just one ``sink'' state where it will
    a.s.~get stuck, see, e.g., (1.5.6) in~\cite{rosenthal2019first}.)
    By the coupling property $X_k\ge \NC_k$, that
    implies $\NC_k=1$ almost surely, which, as explained above, is sufficient for polarization.
\end{proof}

\section{Conclusion}\label{sec:conclusion}
In this paper we studied the phenomenon of polarization in a geometric opinion exchange framework.
We have shown that when we consider opinions as vectors and an interaction between opinions according to a geometric update rule
depending on their overall correlation, then certain update rules obeying biased assimilation will cause polarization (issue radicalization and issue alignment) in a fully connected network. There remain several interesting
avenues for further research.

For example, do there exist
interesting update rules for $d\ge 3$ which do
not always converge to a polarized configuration?
Similarly, are there natural modifications to the model that could
exhibit a more varied behavior?

Another natural research direction is the analysis
of non fully connected networks. 
This seems especially interesting as it could
connect the interaction of the geometry
of interacting opinions (as studied in this paper)
with the geometry of connections between agents.

Finally, we also leave
unaddressed the subjects of the speed of convergence to polarization and which properties affect it,
and the sizes of the two polarized groups for update functions which are not odd.

\appendix
\section{Deferred proofs}
\label{appendix_A}

\begin{proof}[\textbf{Proof of \Cref{lem:transitivity_of_inactivity}}]
    For the first statement, assume that $\eps<1$ (otherwise the statement is trivial)
    and that $A_{ik}\ge 1-\eps$. The other case $A_{ik}\le-(1-\eps)$ follows
    by an easy reduction as in the proof of \Cref{lem:transitivity_of_closeness}.

    By the triangle inequality,
    \begin{equation*}
        \Big|
            \|\vu_j-\vu_k\|-\|\vu_j-\vu_i\|
        \Big|\le\|\vu_i-\vu_k\|=\sqrt{2(1-A_{ik})}\le\sqrt{2\eps}\;.
    \end{equation*}
    On the other hand,
    \begin{equation*}
    \frac{1}{\sqrt{2}}\Big|\|\vu_i-\vu_j\|-\sqrt{2}\Big|
    =\Big|\sqrt{1-A_{ij}}-1\Big|
    \le\left(1-\sqrt{1-|A_{ij}|}\right)
    \le|A_{ij}|\le\sqrt{\eps}\;.
    \end{equation*}
    Putting the two bounds together, 
    $\sqrt{2}-2\sqrt{2\eps}\le\|\vu_j-\vu_k\|\le\sqrt{2}+2\sqrt{2\eps}$,
    which implies $2-16\sqrt{\eps}\le\|\vu_j-\vu_k\|^2\le 2+16\sqrt{\eps}$
    and $|A_{jk}|\le 8\sqrt{\eps}$. The first statement is thus proved.

    \smallskip

    For the second statement, 
    consider an orthonormal coordinate system where \(\vu_i\) has coordinates \((1, 0)\). 
    Let \(x_1, x_2, y_1, y_2 \in \mathbb{R}\) be such that \(\vu_j\) has coordinates 
    \((x_1, y_1)\), and \(\vu_k\) has coordinates \((x_2, y_2)\).

    Now, suppose \(\abs{A_{ij}} \leq \epsilon\) and \(\abs{A_{ik}} \leq \epsilon\). Let's show that \(\abs{A_{jk}} \geq 1 - 2 \epsilon^2\).
    We want to show that for every \(\epsilon\ge 0\), if \(\abs{x_1} \leq \epsilon\) and \(\abs{x_2} \leq \epsilon\), then \(\abs{x_1 x_2 + y_1 y_2} \geq 1 - 2 \epsilon^2\).
    Note that $y_1^2=1-x_1^2\ge 1-\eps^2$ and similarly $y_2^2\ge 1-\eps^2$. 
    In particular, $|y_1y_2|\ge 1-\epsilon^2$
    and $|x_1x_2|\le\epsilon^2$.
    That gives
    \begin{equation*}
        \abs{A_{jk}}
        \geq \abs{y_1 y_2} - \abs{x_1 x_2}
        \geq 1 - \epsilon^2 - \epsilon^2
        = 1 - 2 \epsilon^2 \; . \qedhere
    \end{equation*}
\end{proof}

\begin{proof}[\textbf{Proof of \Cref{lem:transitivity_of_activity}}]
    As before, suppose $\eps_1,\eps_2\le 1$, otherwise the conclusion
    is trivial.
    Suppose also \(A_{ij} \geq \sqrt{\epsilon_1}\) and \(A_{ik} \geq 1 - \epsilon_2\)
    (otherwise we can flip $\vu_j$ and/or $\vu_k$ as necessary, which does
    not change the absolute values $|A_{ij}|,|A_{ik}|,|A_{jk}|$).    
    So, \(1 - A_{ij} \leq 1 - \sqrt{\epsilon_1} \le 1\) and \(1 - A_{ik} \leq \epsilon_2\).
    We have
    \begin{align*}
        A_{jk}
        &= 1 - \frac{\norm{\vu_j - \vu_k}^2}{2} \; , \\
        &\geq 1 - \frac{\left(\norm{\vu_j - \vu_i} + \norm{\vu_i - \vu_k}\right)^2}{2} \; , \\
        &= 1 - \frac{1}{2} \left(\sqrt{2 (1 - A_{ij})} + \sqrt{2(1 - A_{ik})}\right)^2 \; , \\
        &\geq 1 - \left(1 - \sqrt{\epsilon_1} + \epsilon_2 + 2 \sqrt{(1 - \sqrt{\epsilon_1}) \epsilon_2} \right) \; , \\
        &\geq \sqrt{\epsilon_1} - \epsilon_2 - 2 \sqrt{\epsilon_2} \; , \\
        &\geq \sqrt{\epsilon_1} - 3 \sqrt{\epsilon_2} \; . \qedhere
    \end{align*}
\end{proof}
 
\begin{proof}[\textbf{Proof of \Cref{cl:contraction-active}}]
        \(f\) is an active function, so we let \(m > 0\) denote the infimum of \(\abs{f(A)}\) over \([-1, 1]\).
        Recall from \Cref{claim:new_correlation} that the new correlation is given by
        \begin{equation}
            A' = \frac{A + f(A)}{\sqrt{1 + 2 A f(A) + f(A)^2}} \; .
            \label{eq:51}
        \end{equation}
        Let's consider the cases \(\abs{A} > \abs{A_0}\) and \(\abs{A} \leq \abs{A_0}\) separately.
        \paragraph{Case \(\abs{A} > \abs{A_0}\):} We know that \(-1 \leq A \leq 1\) and \(A f(A) \leq \abs{f(A)}\). Therefore         
        \(1 + 2 A f(A) + f(A)^2 \leq \left(1 + \abs{f(A)}\right)^2\). Furthermore, given that \(\abs{A} > \abs{A_0}\), it follows from the definition of an active function that \(\sign(A) = \sign(f(A))\).
        Hence,
        \begin{equation*}
            \abs{A'}
            \geq \frac{\abs{A + f(A)}}{1 + \abs{f(A)}}
            = \frac{\abs{A} + \abs{f(A)}}{1 + \abs{f(A)}}
            = 1- \frac{1-|A|}{1 + |f(A)|} \; .
        \end{equation*}
        From \eqref{eq:51}, \(\sign(A) = \sign(f(A))\) implies \(\sign(A') = \sign(A)\). Thus, proving the claim in this case is equivalent to proving that \(1 - \abs{A'} \leq \beta (1 - \abs{A})\).
        Finally, given that for every \(A\), \(\abs{f(A)} \geq m\), we obtain
        \begin{equation*}
            \abs{A'} 
            \geq 1- \frac{1-|A|}{1 + |f(A)|}
            \geq 1- \frac{1-|A|}{1 + m} \; .
        \end{equation*}
        By letting \(\beta \coloneq 1/(1+m) < 1\), it follows that \(1 - \abs{A'} \leq \beta(1 - \abs{A})\).
        This also implies that \(f\) is \((c, \beta)\)-contractive for \(c = \abs{A_0}\).
    
        \paragraph{Case \(\abs{A} \leq \abs{A_0}\):} Without loss of generality
        suppose \(A_0 \geq 0\) (the case \(A_0 < 0\) is symmetric).
        So, \(\abs{A} \leq \abs{A_0}\) implies \(-A_0 \leq A \leq A_0\). But we need not consider the case \(A = A_0\) as it is not covered by the claim. 
        The argument in this case proceeds as follows. We let \(M \coloneq (A-m)/\sqrt{1-2Am+m^2}\) and we show that \(A' \leq M\). Then we will show that there exists \(\delta > 0\) such that for every \(A \in [-A_0, A_0]\), \(A - \delta \geq M\).
        Indeed that would imply that \(A' \leq A - \delta\) when \(-A_0 \leq A < A_0\).
        By taking \(\beta = (A_0 + 1 - \delta)/(A_0 + 1)\), it would follow that \(A' + 1 \leq A + 1 - \delta < \beta(A + 1)\) because
        \begin{equation*}
            \beta (A + 1) =  \left(1 - \frac{\delta}{A_0 + 1}\right) (A + 1) = A + 1 - \delta \frac{A + 1}{A_0 + 1} > A + 1 - \delta \; ,
        \end{equation*}
        which concludes the proof.
        
        Let's then show that \(A' \leq M\).     
        To that end, fix $-1< A< 1$ and
        let us show that \(A'\) is an increasing function of \(f(A)\).
        Define \(S \coloneq S(x) \coloneq 1 + 2 A x + x^2\).
        We want to show that \(g(x) = (A + x)/\sqrt{S}\) is increasing
        for $x\in\mathbb{R}$.
        Note that \(S > (A+x)^2\ge 0\).
        The result follows directly by considering the derivative of \(g\):
        \begin{equation*}
            g'(x) = \frac{1}{\sqrt{S}} -\frac{(A + x)^2}{S \sqrt{S}}
            = \frac{1}{\sqrt{S}} \left(1 - \frac{(A + x)^2}{S} \right)
            > 0 \; .
        \end{equation*}
        We know that when \(-A_0\le A < A_0\), \(f(A) < 0\), so \(f(A) \leq -m\). Given that \(A'\) increases with \(f(A)\) we get
        \begin{equation*}
            A' \leq \frac{A - m}{\sqrt{1 - 2 A m + m^2}} = M \; .
        \end{equation*}
        We will now show that $ M = (A-m)/\sqrt{1-2Am+m^2}<A$ for $-A_0\le A\le A_0$.
        Suppose for the sake of contradiction that \((A - m)\sqrt{1 - 2 A m + m^2} \geq A\).
        This implies
        \begin{equation}
            \label{eq:A-m}
            A - m \geq A \sqrt{1 - 2 A m + m^2} \; .
        \end{equation}
        If \(-A_0 \leq A \le 0\), then
        \((A - m)^2 \leq A^2 ( 1 - 2Am + m^2)\) and so \(m^2 - 2 A m \leq A^2 (m^2 - 2Am)\). Given that \(m^2 - 2 A m > 0\), we obtain the contradiction \(A^2 \geq 1\).
        Similarly, if \(0 < A \le A_0\), then
        \((A - m)^2 \geq A^2 (1 - 2Am + m^2)\) and so \(m^2 - 2 A m \geq A^2 (m^2 - 2Am)\). In this case \(m^2 - 2 A m < 0\) because \(A - m > 0\) from \eqref{eq:A-m}. The contradiction \(A^2 \geq 1\) follows again.
        We conclude that
        \(A - (A - m)/\sqrt{1 - 2 A m + m^2}\) is a positive continuous function of \(A\) on the closed interval \([-A_0, A_0]\). Therefore, it has a minimum \(\delta > 0\) on that interval.
        We conclude that
        \(M \leq A - \delta\) when \(-A_0 \leq A < A_0\). The claim follows.
\end{proof}

\begin{proof}[\textbf{Proof of \Cref{claim:3}}]
    For further use, we note that
    \begin{subequations}
    \begin{equation}
        \label{eq:cosine_upper}
        \forall x, \; 1 - \cos x \leq \frac{1}{2} x^2 \; .
    \end{equation}
    \begin{equation}
        \label{eq:cosine_lower}
        \exists c > 0 \, : \, \forall 0 \le x \le c, \; 1 - \cos x \geq \frac{x^2}{2} - x^4 \; .
    \end{equation}
    \end{subequations}

    First let us prove the bound on $\gamma'$ for
    $\gamma\le\gamma_0$, where we choose
    $\gamma_0>0$ small enough.
    Let \( \epsilon:=1-\abs{A}\), \( \epsilon':=1-\abs{A'}\). By Claim \ref{claim:contraction_stable}, there exists \(0 < c < 1\) such that \(\epsilon' \le c \epsilon\). Note that \(\gamma = \arccos(|A|)\) and \(\gamma' = \arccos(|A'|)\).
    
    From the bounds on cosine (Eq. \ref{eq:cosine_upper}, \ref{eq:cosine_lower}), we have
    \begin{equation*}
        \epsilon = 1 - \cos{\gamma} \leq \frac{1}{2}\gamma^2
        \quad \text{ and } \quad
        \epsilon' = 1 - \cos{\gamma'} \geq \frac{1}{2}(\gamma')^2 - (\gamma')^4 \; .
    \end{equation*}
    Furthermore, for every \(\delta > 0\), there
    exists $x_0$ such that if \(x < x_0\), then \(x^4 < \delta \frac{1}{2}x^2\).
    Therefore, if \(\gamma' \le \gamma\le x_0\), then 
    \begin{equation*}
        \epsilon' \ge \frac{(1 - \delta)}{2}(\gamma')^2 \; .    
    \end{equation*}
    We then have
    \begin{equation*}
        \frac{(1 - \delta)}{2}(\gamma')^2 \le \epsilon' \le c \epsilon \leq \frac{c}{2}\gamma^2 \; , 
    \end{equation*}
    and so \(\gamma' < \sqrt{c/(1 - \delta)} \gamma\). By choosing \(\delta < 1 - c\), and setting \(c_0 = \sqrt{c/(1 - \delta)} < 1\), we conclude
    the case $\gamma\le\gamma_0$.

    Finally, we also need to obtain the bound for 
    \(\gamma \in [\gamma_0, \frac{\pi}{2}-\delta]\). 
    First, assume $\gamma=\alpha$.
    Note that \Cref{cor:A_is_non_decreasing} implies that \(\gamma' < \gamma\) when \(\gamma \in (0, \frac{\pi}{2})\).
    Furthermore, given that \(A'\) is a continuous function of \(A\), we have that \(\gamma{'}/\gamma<1\) is a continuous function of \(\gamma\in [\gamma_0,\pi/2-\delta]\).
    So, there exists \(0 < M < 1\) such that if \(\gamma \in [\gamma_0, \frac{\pi}{2}-\delta]\), then \(\gamma'/\gamma \leq M\).
    If $\gamma=\pi-\alpha$, the argument is entirely
    symmetric.
    
    We conclude that over the whole interval \([0, \frac{\pi}{2}-\delta]\) it holds \(\gamma' \leq \max(c_0, M) \gamma\).
\end{proof}

\paragraph{Acknowledgements} This research project was initiated as a part of the DAAD activity
in Dortmund organized by Amin-Coja Oghlan and Nicola Kistler and their research groups.
We are grateful to the referees
for their extensive and helpful comments.

\printbibliography

\end{document}